\documentclass[12pt]{article}

\usepackage{graphics,graphicx,fullpage,multirow}
\usepackage{amsmath,amssymb,verbatim,natbib}
\usepackage[dvipsnames,usenames]{color}
\usepackage{subfig}
\usepackage{color}

\newtheorem{proposition}{Proposition}

\newtheorem{defin}{\bf Definition}
\newenvironment{definition}{\begin{defin}\rm}{\end{defin}}
\newenvironment{proof}{\noindent{\it Proof.}}{$\diamond$}

\def\ga{\mbox{Ga}}
\def\be{\mbox{Be}}

\def\no{\mbox{N}}

\def\dir{\mbox{Dir}}

\def\E{\mbox{E}}
\def\V{\mbox{Var}}

\def\P{\mbox{P}}

\def\d{\mbox{d}}
\def\data{\mbox{data}}

\def\br{{\bf r}}

\def\bx{{\bf x}}
\def\by{{\bf y}}

\def\bI{{\bf I}}
\def\bN{{\bf N}}

\def\bU{{\bf U}}
\def\bX{{\bf X}}
\def\bY{{\bf Y}}
\def\bZ{{\bf Z}}

\newcommand{\balpha}{\boldsymbol{\alpha}}

\newcommand{\bepsilon}{\boldsymbol{\epsilon}}
\newcommand{\bgamma}{\boldsymbol{\gamma}}

\newcommand{\bmu}{\boldsymbol{\mu}}
\newcommand{\bpi}{\boldsymbol{\pi}}

\newcommand{\Ree}{{\rm I}\!{\rm R}}

\newcommand{\AC}{\mathcal{A}}
\newcommand{\BB}{\mathcal{B}}

\newcommand{\YY}{\mathcal{Y}}

\newcommand{\PT}{\mbox{PT}}
\newcommand{\PPT}{\mbox{PPT}}

\begin{document}

\baselineskip=24pt

\title{{\bf Projected P\'olya Tree}}
\author{
  {\sc  Luis Nieto-Barajas$^1$} and
  {\sc  Gabriel N\'u\~nez-Antonio$^2$} \\
  {\sl {\small $^1$Department of Statistics, ITAM, Mexico}} \\ 
  {\sl {\small $^2$Department of Mathematics, UAM-I, Mexico}}}
\date{}

\maketitle

\begin{abstract}
One way of defining probability distributions for circular variables (directions in two dimensions) is to radially project probability distributions, originally defined on $\mathbb{R}^2$, to the unit circle. Projected distributions have proved to be useful in the study of circular and directional data. Although any bivariate distribution can be used to produce a projected circular model, these distributions are typically parametric. In this article we consider a bivariate P\'olya tree on $\mathbb{R}^2$ and project it to the unit circle to define a new Bayesian nonparametric model for circular data. We study the properties of the proposed model, obtain its posterior characterisation and show its performance with simulated and real datasets. 
\end{abstract}


\section{Introduction}
\label{sec:intro}

Directional data arise from the observation of unit vectors in $k$-dimensional space and, consequently, they can be represented through $k-1$ angles. Thus, the sample space associated with this type of data is the $k$-dimensional unit sphere, $\mathbb{S}^k$. The most common case is when $k=2$ producing so called circular data. This type of data is especially common in biology, geophysics, meteorology, ecology and environmental sciences. Specific applications include the study of wind directions, orientation data in biology, direction of birds migration, directions of fissures propagation in concrete and other materials, orientation of geological deposits, and the analysis of mammalian activity patterns in ecological reserves, among others. 
For a survey on the area, the reader is referred to classic literature, e.g. \cite{mardia:72}, \cite{fisher:95}, \cite{mardia&jupp:00} and \cite{jamma&gupta:01}. For a more recent overview of applications of circular data analysis in ecological and environmental sciences see \cite{arnold&gupta:06} and \cite{lee:10}.

In recent years the development of statistical methods to analyse directional data has had a new interest. \cite{presnell&al:98} considered the case of projected linear models, \cite{delia&al:01} studied longitudinal circular data, and \cite{paine&al:18} introduced an elliptically symmetric angular Gaussian distribution for the study of directional data on $\mathbb{S}^k$. 

While there are several ways to define probability distributions for directional random vectors, one of the simplest ways to generate distributions on $\mathbb{S}^k$ is to radially project probability distributions originally defined on $\mathbb{R}^k$. A directional distribution which has received a lot of attention is the special case where the distribution to project is a $k$-variate Normal distribution; in this case it is said the corresponding directional variable has a \textit{projected Normal distribution} \citep[e.g.][]{mardia&jupp:00}. Within a Bayesian context, this model has been studied by \cite{nunez&gutierrez:05} and \cite{wang&gelfand:13} for the circular case, and \cite{hernandez&al:17} for the $k$-dimensional case.

Although any bivariate distribution can be used to produce a projected circular model, these distribution are typically parametric. However, in real situations it may be preferable to consider semiparametric or nonparametric models as an alternative to properly describe the behaviour of this kind of data. In a classical context, nonparametric modelling for circular data has been typically carried out using a circular kernel density such as the von Mises distribution \citep[e.g.][]{fisher:89,oliveira&al:14} or via nonnegative trigonometric sums \citep{fernandez&gregorio:16}. Within a Bayesian nonparametric approach, a unimodal and symmetric density \citep{brunner&lo:94}, Dirichlet processes mixtures (DPM) of  von Mises distributions \citep{ghosh&al:03}, DPM of projected normal distributions \citep{nunez&al:15} and more recently mixtures of basis of trigonometric polynomials \citep{binette&guillotte:18} have been proposed. Other semiparametric approaches are mixtures of triangular distributions \citep{mcvinish&mengersen:08} and log-spline distributions \citep{ferreira&al:08}. In this work we consider a bivariate P\'olya tree and project it to the unit circle to produce a projected P\'olya tree model. This new Bayesian nonparametric model for circular data will be shown to be competitive with respect to other Bayesian nonparametric proposals with the advantage of the simplicity in carrying out posterior inference. 

The rest of the paper is organized as follows. In Section \ref{sec:pt} we set the notation and present basic ideas about univariate P\'{o}lya trees. In Section \ref{sec:model} we introduce the projected P\'olya tree prior and study its properties. In Section \ref{sec:post} we describe how to perform posterior inference via a data augmentation technique. We illustrate the performance of our proposal in Section \ref{sec:numerical} via a simulation study and the analysis of a real data set. We conclude with some remarks in Section \ref{sec:concl}.

Before proceeding we introduce notation: $\ga(\alpha,\beta)$ denotes a gamma density with mean $\alpha/\beta$; $\no(\mu,\tau)$ denotes a normal density with mean $\mu$ and precision $\tau$; $\be(\alpha_1,\alpha_2)$ denotes a beta density with mean $\alpha_1/(\alpha_1+\alpha_2)$; $\no_2(\bmu,\Sigma)$ denotes a bivariate normal density with mean vector $\bmu$ and precision matrix $\Sigma$; and $\dir(\balpha)$ denotes a dirichlet density with parameter vector $\balpha$.

\section{P\'{o}lya Tree}
\label{sec:pt}

In this section we recall the definition of a univariate P\'olya tree and set notation. Consider $(\mathbb{R},\BB)$, the measurable space with $\mathbb{R}$ the real line and $\BB$ the Borel sigma algebra of subsets of $\mathbb{R}$. We require a binary partition tree, which using notation from \cite{nieto&mueller:12}, is denoted by  $\Pi=\{B_{mj}: m\in\mathbb{N}, j=1,\dots,2^m\}$, where the index $m$ specifies the level of the tree and $j$ the location of the partitioning subset within the level. In general, at level $m$, the set $B_{mj}$ splits into two disjoint sets $(B_{m+1,2j-1},B_{m+1,2j})$. For every set $B_{mj}$ there is associated a branching probability $Y_{mj}$ such that $Y_{m+1,2j-1} = F(B_{m+1,2j-1} \mid B_{mj})$, and $Y_{m+1,2j}=1-Y_{m+1,2j-1} = F(B_{m+1,2j} \mid B_{mj})$, where $F$ will be used to denote a cumulative distribution function (cdf) or a probability measure indistinctively.  

\begin{definition}
\label{def:uniPT}
\citep{lavine:92}. Let $\AC_m=\{\alpha_{mj},\, j=1,\ldots,2^m\}$ be non-negative real numbers, $m=1,2,\ldots,$ and let $\AC=\bigcup \AC_m$. A random probability measure $F$ on $(\mathbb{R},\BB)$ is said to have a P\'olya tree prior with parameters $(\Pi,\AC)$, if for $m=1,2,\ldots$ there exist random variables $\YY_m=\{Y_{m,2j-1}\}$ for $j=1,\ldots,2^{m-1}$, such that the following hold: 
\begin{enumerate}
\item All the random variables in $\YY=\cup_m\{\YY_m\}$ are independent. 
\item For every $m=1,2,\ldots$ and every $j=1,\ldots,2^{m-1}$, $Y_{m,2j-1}\sim\be(\alpha_{m,2j-1},\alpha_{m,2j})$.
\item For every $m=1,2,\ldots$ and every $j=1,\ldots,2^m$ 
$$F(B_{mj})=\prod_{k=1}^m Y_{m-k+1,j_{m-k+1}^{(m,j)}},$$
where $j_{k-1}^{(m,j)}=\lceil j_k^{(m,j)}/2 \rceil$ is a recursive decreasing formula, whose initial value is $j_{m}^{(m,j)}=j$, that locates the set $B_{mj}$ with its ancestors upwards in the tree. $\lceil\cdot\rceil$ denotes the ceiling function, and $Y_{m,2j}=1-Y_{m,2j-1}$ for $j=1,\ldots,2^{m-1}$.
\end{enumerate}
\end{definition}

A P\'olya tree prior can be centred around a parametric probability measure $F_0$. The simplest way \citep{hanson&johnson:02} consists of matching the partition with the dyadic quantiles \citep{ferguson:74} of the desired centring measure and keeping $\alpha_{mj}$ constant within each level $m$. More explicitly, at each level $m$ we take  
\begin{equation}
\label{bmj}
B_{mj}=\left(F_0^{-1}\left(\frac{j-1}{2^m}\right),F_0^{-1}\left(\frac{j}{2^m}\right)\right],
\end{equation}
for $j=1,\ldots,2^m$, with $F_0^{-1}(0)=-\infty$ and $F_0^{-1}(1)=\infty$. Note that for $j=2^m$, $B_{mj}$ is defined open on both sides. If we further take $\alpha_{mj}=\alpha_m$ for $j=1,\ldots,2^m$ we get $\E\{F(B_{mj})\}=F_0(B_{mj})$.

In particular, we take $\alpha_{mj}=\alpha\rho(m)$, so that the parameter $\alpha$ can be interpreted as a precision parameter of the P\'olya tree \citep{walker&mallick:97}, and the function $\rho$ controls the speed at which the variance of the branching probabilities moves down in the tree. As suggested by \cite{watson&al:17} we take $\rho(m)=m^\delta$ with $\delta>1$ to ensure the process $F$ is absolutely continuous  \citep{kraft:64}.

\section{Main model}
\label{sec:model}

\subsection{Bivariate P\'olya tree}

In this section we generalize the univariate P\'olya tree to a bivariate one. Let $(\mathbb{R}^2,\BB^2)$ be our measurable space. There are several ways of defining and denoting the nested partition $\Pi$ \citep{paddock:02,hanson:06,jara&al:09,filippi&holmes:17}. For simplicity, we define the partition as the cross product of univariate partitions and use the notation of \cite{nieto&mueller:12} presented in Section \ref{sec:intro}. In other words, $\Pi=\{B_{m,j,k}\}$ such that $B_{m,j,k}=B_{m,j}\times B_{m,k}$, for $j,k=1,\ldots,2^{m}$ and $m=1,2,\ldots$. The index $m$ denotes the level of the tree and the pair $(j,k)$ locates the partitioning subset within the level. In general, the set $B_{m,j,k}$ splits into four disjoint subsets $(B_{m+1,2j-1,2k-1},B_{m+1,2j-1,2k},B_{m+1,2j,2k-1},B_{m+1,2j,2k})$. At each level $m$ we will have a partition of size $4^m$. We associate random branching probabilities $Y_{m,j,k}$ with every set $B_{m,j,k}$ such that, for example, $Y_{m+1,2j-1,2k-1}=F(B_{m+1,2j-1,2k-1}\mid B_{m,j,k})$, where again $F$ denotes a cdf or a probability measure, indistinctively.

\begin{definition}
\label{def:biPT}
Let $\AC=\{\alpha_{m,j,k}\}$, $j,k=1,\ldots,2^m$, $m=1,2,\ldots$ be a set of nonnegative real numbers. A random probability measure $F$ on $(\mathbb{R}^2,\BB^2)$ is said to have a bivariate P\'olya tree prior with parameters $(\Pi,\AC)$ if there exists random vectors $\bY_{m,j,k}=(Y_{m+1,2j-1,2k-1},Y_{m+1,2j-1,2k},Y_{m+1,2j,2k-1},Y_{m+1,2j,2k})$ such that the following hold: 
\begin{enumerate}
\item All random vectors $\bY_{m,j,k}$, $j,k=1,\ldots,2^m$ and $m=0,1,2,\ldots$ are independent
\item For every $m=0,1,\ldots$ and every $j,k=1,\ldots,2^m$, $\bY_{m,j,k}\sim\dir(\balpha_{m,j,k})$, where $\balpha_{m,j,k}=(\alpha_{m+1,2j-1,2k-1},\alpha_{m+1,2j-1,2k},\alpha_{m+1,2j,2k-1},\alpha_{m+1,2j,2k})$
\item For every $m=1,2,\ldots$ and every $j,k=1,\ldots,2^m$, $$F(B_{m,j,k})=\prod_{l=1}^m Y_{m-l+1,j_{m-l+1}^{m,j,k},k_{m-l+1}^{m,j,k}},$$ where $j_{l-1}^{(m,j,k)}=\left\lceil\frac{j_l^{(m,j,k)}}{2}\right\rceil$ and $k_{l-1}^{(m,j,k)}=\left\lceil\frac{k_l^{(m,j,k)}}{2}\right\rceil$ are recursive decreasing formulae, whose initial values are $j_m^{(m,j,k)}=j$ and $k_m^{(m,j,k)}=k$, that locate the set $B_{m,j,k}$ with its ancestors upwards in the tree. 
\end{enumerate}
\end{definition}
Note that in the previous definition $\bY_{0,1,1}=(Y_{1,1,1},Y_{1,1,2},Y_{1,2,1},Y_{1,2,2})$ and $\balpha_{0,1,1}=(\alpha_{1,1,1},$ $\alpha_{1,1,2},\alpha_{1,2,1},\alpha_{1,2,2})$ are the vectors associated to the partition elements at level $m=1$.

It is desired to center the bivariate P\'olya tree around a parametric probability measure $F_0$. For simplicity, let us assume that $F_0(x_1,x_2)=F_{1_0}(x_1)F_{2_0}(x_2)$. Non-independence $F_0$ could also be considered but a suitable transformation of the partition sets $B_{m,j,k}$ would be required \citep[e.g.][]{jara&al:09}. Therefore, we proceed as in the univariate case by matching the partition $B_{m,j,k}=B_{m,j}\times B_{m,k}$ with the dyadic quantiles of the marginals $F_{1_0}$ and $F_{2_0}$, i.e.,
\begin{equation}
\label{bmj2}
B_{m,j}=\left(F_{1_0}^{-1}\left(\frac{j-1}{2^m}\right),F_{1_0}^{-1}\left(\frac{j}{2^m}\right)\right]\quad\mbox{and}\quad B_{m,k}=\left(F_{2_0}^{-1}\left(\frac{k-1}{2^m}\right),F_{2_0}^{-1}\left(\frac{k}{2^m}\right)\right],
\end{equation}
for $j,k=1,\ldots,2^m$. We further define $\balpha_{m,j,k}=(\alpha\rho(m+1),\ldots,\alpha\rho(m+1))$ where $\alpha>0$ is the precision parameter and $\rho(m)=m^\delta$ with $\delta>1$ to define an absolutely continuous bivariate P\'olya tree. It is not difficult to prove that a bivariate P\'olya tree, defined in this way, satisfies $\E\{F(B_{m,j,k})\}=F_0(B_{m,j,k})=1/4^m$. 

In practice we need to stop partitioning the space at a finite level $M$ to define a finite tree process. At the lowest level $M$, we can spread the probability within each set $B_{M,j,k}$ according to $f_0$, the density associated to $F_0$. In this case the random probability measure defined will have a bivariate density of the form
\begin{equation}
\label{eq:dPT}
f(\bx)=\left\{\prod_{m=1}^M Y_{m,j_m^{(x_1)},k_m^{(x_2)}}\right\}4^M f_0(\bx),
\end{equation}
where $\bx'=(x_1,x_2)\in\mathbb{R}^2$, and with $(j_m^{(x_1)},k_m^{(x_2)})$ identifying the set at level $m$ that contains $\bx$. This maintains the condition $\E(f)=f_0$. We denote a finite bivariate P\'olya tree process as $\PT_M(\alpha,\rho,F_0)$. By taking $M\to\infty$ we recover the (infinite) bivariate P\'olya tree of Definition \ref{def:biPT}.

It is well known \citep[e.g.][]{lavine:92} that univariate and multivariate P\'olya tree densities are discontinuous at the boundaries of the partitions. To overcome this feature, an extra mixture with respect to the parameters of the centering measure is imposed, that is, $f_0$ is replaced by $f_0(\cdot\mid\eta)$ and a prior $f(\eta)$ is placed to induce smoothness.

\subsection{Projected tree}

We are now in a position to construct the projected P\'olya tree. Let us assume a bivariate random vector $ \bX'=(X_1,X_2)$ such that $\bX\mid f \sim f$ and $f$ is given in \eqref{eq:dPT}. We project the random vector $\bX$ to the unit circle by defining $\bU=\bX/||\bX||$. Alternatively, we can work with the polar coordinates transformation $(X_1,X_2)\rightarrow (\Theta,R)$, where $\Theta$ is the angle and $R=||\bX||$ is the resultant length of the vector in the plane. The inverse transformation becomes $X_1=R\cos\Theta$ and $X_2=R\sin\Theta$. Thus, the corresponding Jacobian is $J=R$. Then, the induced marginal density for the angle $\Theta$ has the form
\begin{equation}
\label{eq:PPT}
f(\theta)=\int_0^\infty\left\{\prod_{m=1}^M Y_{m,j_m^{(r\cos\theta)},k_m^{(r\sin\theta)}}\right\}4^M f_0(r\cos\theta,r\sin\theta)\,|J|\,\d r.
\end{equation}
We will refer to the density $f(\theta)$, given in \eqref{eq:PPT}, as the \textit{projected P\'olya tree} and will be denoted by $\PPT_M(\alpha,\rho,f_0)$. 

In contrast to P\'olya tree densities, the projected P\'olya tree \eqref{eq:PPT} is not discontinuous at the boundaries of the partitions. The reason for the smoothing effect relies on the marginalisation when passing from the joint density $f(\theta,r)$ to the marginal $f(\theta)$, which can also be seen as a mixture of the form $f(\theta)=\int f(\theta\mid r)f(r)\d r$. A specific angle $\theta_0$ might come from many points in $\Ree^2$ defined by different resultants in polar coordinates say $(\theta_0,r_l)$, $l=1,2,\ldots$. Each of these points might belong to different partition sets, which are added (integrated) in the marginalisation. Therefore, no extra mixing is required to produce smooth densities. 

In particular, we can center our projected P\'olya tree on the projected normal distribution, considered by \cite{nunez&gutierrez:05}, by taking $f_0(\bx)=\no_2(\bx\mid\bmu,\bI)$, that is, a bivariate normal density with mean vector $\bmu'=(\mu_1,\mu_2)$ and precision the identity matrix $\bI$. In this case, the projected P\'olya tree becomes 
\begin{align}
\nonumber
f(\theta)=\int_0^\infty & \left\{\prod_{m=1}^M Y_{m,j_m^{(r\cos\theta)},k_m^{(r\sin\theta)}}\right\}4^M (2\pi)^{-1}e^{-\frac{1}{2}\bmu'\bmu}\ r \\
\label{eq:PPTn}
&\times\exp\left[-\frac{1}{2}\left\{r^2-2r\left(\mu_1\cos\theta+\mu_2\sin\theta\right)\right\}\right]I_{(0,2\pi]}(\theta)\d r.
\end{align}

In general, the marginal density $f(\theta)$ does not have an analytic expression. However, it can be computed numerically via quadrature. Say, if $0=r^{(0)}<r^{(1)}<\cdots<r^{(L)}<\infty$ is a partition of the positive real line, then 
$$f(\theta)\approx\sum_{l=1}^L f(r^{(l)}\cos\theta,r^{(l)}\sin\theta)\,|J|\left(r^{(l)}-r^{(l-1)}\right),$$
where $f(\cdot,\cdot)$ is given in \eqref{eq:dPT}. Alternatively, $f(\theta)$ can also be approximated via Monte Carlo.

Densities for circular variables are periodic, that is $f(\theta+2\pi)=f(\theta)$, therefore moments are not defined in the usual way \citep[e.g.][]{rao&umbach:10}. Instead, the $p^{th}$ trigonometric moment of a random variable $\Theta$ is a complex number of the form $\varphi_p=\E(e^{ip\Theta})=a_p+ib_p$, for any integer $p$, where $a_p=\E(\cos p\Theta)$ and $b_p=\E(\sin p\Theta)$. The mean direction of $\Theta$, $\nu_\theta$, and a concentration measure around the mean, $\varrho_\theta$, are defined as 
\begin{equation}
\label{eq:moments}
\nu_\theta=\arctan(b_1/a_1),\quad\varrho_\theta=\sqrt{a_1^2+b_1^2}, 
\end{equation}
where $\varrho_\theta\in[0,1]$. A value of $\varrho_\theta$ close to one means that $\Theta$ is highly concentrated around its mean $\nu_\theta$, whereas a value of $\varrho_\theta$ close to zero means that $\Theta$ is highly disperse. 

To illustrate what the paths of the projected P\'olya tree look like, we consider the model centred around the projected normal, as in \eqref{eq:PPTn}, with four levels of the partition ($M=4$), a precision parameter $\alpha=1$, a function $\rho(m)=m^{\delta}$ with $\delta=1.1$, and different values of $\bmu$. For each setting we sampled ten paths (densities) from the model. The marginal density of $\theta$ is approximated numerically with $L=100$ points.

Figure \ref{fig:priorppt1} contains four panels which correspond to $\bmu'=(0,1)$ (top left), $\bmu'=(1,0)$ (top right), $\bmu'=(0,-1)$ (bottom left) and $\bmu'=(-1,0)$ (bottom right). These values of $\bmu$ represent specific values of the bivariate normal mean in locations around the unit circle at $\pi/2$, $2\pi$, $3\pi/2$, and $\pi$, respectively. We first note that the densities are connected in the sense that the value at $\theta=0$ coincides with the value at $\theta=2\pi$, as they should be. Within each panel we see the diversity of the paths, most of them present a multimodal behaviour. However, the predominant modes are located around the directions of the $\bmu$'s for each of the graphs in the four panels. 

In a different scenario, we move the bivariate normal mean away from the origin to see the impact in the projected tree. This is presented in Figure \ref{fig:priorppt2} that contains four panels which correspond to $\bmu'=(0,0)$ (top left), $\bmu'=(1,1)$ (top right), $\bmu'=(2,2)$ (bottom left) and $\bmu'=(5,5)$ (bottom right). The first panel corresponds to the projected tree centred around the uniform density, obtained when $\bmu'=(0,0)$, however the simulated paths show a high variability around the centring density. As we move away from the origin (second to fourth panels) two things happen, there starts to appear a dominant mode around $\pi/4$, and the variability of the paths highly decreases. This is an interesting finding because in P\'olya trees the variability is entirely controlled by the parameters $\alpha$ and $\rho(\cdot)$ \citep[e.g.][]{hanson:06} and not by the centring measure. What we are seeing in this Figure \ref{fig:priorppt2} is that the variability of the paths in this projected P\'olya tree is also controlled by the centring measure and specifically by its location vector's norm. In other words, the location parameter of the centring measure not only controls the shape of the densities but also the variability of the paths.  

For the eight values of $\bmu$ studied, in Figure \ref{fig:priormoments} we also show the prior distribution of the mean $\nu_\theta$ and the prior distribution of the concentration $\varrho_\theta$, given in \eqref{eq:moments}. We based our prior distributions on 500 simulated paths of the corresponding projected P\'olya tree. For varying $\bmu$ around the unit circle, we see that $\nu_\theta$ (top left panel) has a symmetric distribution with low variability and locations that move at $\pi/2$, $2\pi$, $3\pi/2$ (or $-\pi/2$), and $\pi$, respectively. However, the concentration parameter $\varrho_\theta$ (top right panel) has practically the same symmetric distribution for the different values of $\bmu$ around the unit circle. On the other hand (in the bottom left panel), the distribution of $\nu_\theta$ for $\bmu'=(0,0)$ is uniform, whereas this distribution is increasingly concentrated around $\pi/4$, when $\bmu$ moves from $(2,2)$ to $(5,5)$. Finally, the concentration parameter $\varrho_\theta$ (bottom right panel) has a dispersed distribution around 0.4, for $\bmu'=(0,0)$, and moves to distributions less dispersed and locations that increase its values closer to one, when $\bmu$ moves from $(2,2)$ to $(5,5)$. 

A typical concern in Bayesian nonparametric priors is posterior consistency of the model. That is, we want to be sure that the posterior distribution concentrates around (weak) neighbours of a particular density, say $f^*(\theta)$, when the sample size $n$ goes to infinity. \cite{barron:98} proved that this property if satisfied as long as the prior $f$ puts positive mass around a Kullback-Leibler neighbour of $f^*$. That is, we want $\P\{\mbox{KL}(f^*,f)<\epsilon\}>0$ where $\mbox{KL}(f^*,f)=\int\log\left\{f^*(x)/f(x)\right\}f^*(x)\d x$. The following result states conditions for this to happen. 

\begin{proposition}
\label{prop:consistency}
Let $f\sim\PPT(\alpha,\rho,f_0)$ as in \eqref{eq:PPT} with $M\to\infty$. Let $f^*(\theta)$ be an arbitrary density such that $\mbox{KL}(f^*,f_0)<\infty$. Then, if $\sum_{m=1}^\infty\rho(m)^{-1/2}<\infty$, as $n\to\infty$ $f$ achieves weak posterior consistency. 
\end{proposition}
\begin{proof}
The idea of the proof is to prove posterior consistency for a bivariate density $f^*(\theta,r)=f^*(\theta)f^*(r)$, where $f^*(r)$ is an arbitrary density for a latent resultant $r$. Following proof of Theorem 3.1 in \cite{ghosal&al:99}, by the martingale convergence theorem, there exists a collection of numbers $\{y_{m,j,k}\}$ in $[0,1]$ such that with probability one
$f^*(\theta,r)=\lim_{M\to\infty}r\prod_{m=1}^M\left\{4y_{m,j_m^{(r\cos\theta)},k_m^{(r\sin\theta)}}\right\}$. Now, by \eqref{eq:dPT} and for $M\to\infty$ we have that $f(\theta,r)=\lim_{M\to\infty}r\prod_{m=1}^M\left\{4Y_{m,j_m^{(r\cos\theta)},k_m^{(r\sin\theta)}}\right\}$. The proof continues analogous to Ghosal's. However, they show that $\eta(k)=\E\{|\log(2U_k)|\}=O(k^{-1/2})$ where $U_k\sim\be(k,k)$. In our case we need to prove that $\E\{|\log(4V_k)|\}=O(k^{-1/2})$ where $V_k\sim\be(k,3k)$. We note that $\E(2U_k)=1$ and $\V(2U_k)=1/(2k+1)$ and that $\E(4V_k)=1$ and $\V(4V_k)=3/(4k+1)$. Since both variances have the same rate of decay, our requirement is also true, proving the result.
\end{proof}

In other words, what Proposition \ref{prop:consistency} states is that, if $\rho(m)=m^\delta$, we need $\delta>2$ to satisfy the posterior consistency property. On the other hand, \cite{watson&al:17} suggest a $\delta$ close to one, say $\delta=1.1$, to maximise the dispersion of a finite tree and make the prior less informative. Other authors have made inference on $\delta$ by assigning it a hyper-prior distribution \citep{hanson&al:08}. Moreover, according to \cite{hanson&johnson:02}, if there are no ties in the data, a finite tree provides the same inference as an infinite tree, as long as $M$ is large enough. Since we will be using finite trees we will follow \cite{watson&al:17}'s suggestion.

\section{Posterior inference}
\label{sec:post}

Let $\Theta_1,\Theta_2,\ldots,\Theta_n$ be a sample of size $n$ such that $\Theta_i\mid f\sim f$, independently, and $f\sim\PPT_M(\alpha,\rho,f_0)$, as in \eqref{eq:PPT}. We consider a data augmentation approach \citep{tanner:91} by defining latent resultant lengths $R_1,R_2,\ldots,R_n$ such that $(\Theta_i,R_i)$ define the polar coordinate transformation of the bivariate $(X_{1i},X_{2i})$ on the plane, for $i=1,\ldots,n$.

Then, the likelihood for $\{\bY_{m,j,k}\}$, $j,k=1,\ldots,2^m$ and $m=0,1,\ldots,M-1$, given the extended data, is 
$$\mbox{lik}(\bY\mid\data)=\prod_{i=1}^n\prod_{m=1}^M Y_{m,j_m^{(r_i\cos\theta_i)},k_m^{(r_i\sin\theta_i)}}=
\prod_{m=1}^M\prod_{j=1}^{2^m}\prod_{k=1}^{2^m}Y_{m,j,k}^{N_{m,j,k}},$$ 
where $N_{m,j,k}=\sum_{i=1}^n I(r_i\cos\theta_i\in B_{m,j})I(r_i\sin\theta_i\in B_{m,k})$. 

Recalling from Definition \ref{def:biPT} that the prior distribution of the vectors $\bY_{m,j,k}$ is Dirichlet with parameter $\balpha_{m,j,k}$, and noting that the likelihood is conjugate with respect to this prior, then the posterior distribution for the branching probability vectors is
\begin{equation}
\label{eq:posty}
\bY_{m,j,k}\mid\data\sim\dir(\balpha_{m,j,k}+\bN_{m,j,k}).
\end{equation}
where $\bN_{m,j,k}=(N_{m+1,2j-1,2k-1},N_{m+1,2j-1,2k},N_{m+1,2j,2k-1},N_{m+1,2j,2k})$. 

Remember that this posterior depends on an extended version of the data. The latent resultant lengths $R_i$'s have to be sampled from their corresponding posterior predictive distribution which is simply 
\begin{equation}
\label{eq:postr}
f(r_i\mid\bY,\theta_i)\propto\left\{\prod_{m=1}^M Y_{m,j_m^{(r_i\cos\theta_i)},k_m^{(r_i\sin\theta_i)}}\right\}f_0(r_i\cos\theta_i,r_i\sin\theta_i)\,r_i,
\end{equation}
for $i=1,\ldots,n$.

With equations \eqref{eq:posty} and \eqref{eq:postr} we can implement a Gibbs sampler \citep{smith&roberts:93}. Sampling from \eqref{eq:posty} is straightforward and to sample from \eqref{eq:postr} we will require a Metropolis-Hastings (MH) step \citep{tierney:94}. For this we propose a random walk proposal distribution such that, dropping the index $i$, at iteration $(t+1)$ we sample $r^*$ from $\ga(\kappa,\kappa/r^{(t)})$ and accept it with probability $$\pi(r^*,r^{(t)})=\frac{f(r^*\mid\bY,\theta)\ga(r^{(t)}\mid\kappa,\kappa/r^*)}{f(r^{(t)}\mid\bY,\theta)\ga(r^*\mid\kappa,\kappa/r^{(t)})}.$$
This latter is truncated to the interval $[0,1]$. The parameter $\kappa$ is a tuning parameter, chosen appropriately to produce good acceptance probabilities. Alternatively, this parameter can also be chosen via adaptive MH \citep[e.g.][]{haario&al:01}. For the examples considered here we used $\kappa=0.5$ and obtained acceptance rates between $0.2$ and $0.4$, which according to \cite{robert&casella:10} are optimal.

Posterior inference of our PPT has been implemented in the \textsf{R}-package \texttt{PPTcirc} \citep{perez&nieto:20}. 

As a measure of goodness of fit, for each prior scenario we will compute the logarithm of the pseudo marginal likelihood (LPML), originally suggested by \cite{geisser&eddy:79}. In particular, we will choose the best value for $\bmu$ and $\alpha$ by comparing these measures.

\section{Numerical studies}
\label{sec:numerical}

\subsection{Simulation study}

We consider a model that is based on the projection of a bivariate normal mixture with four components. Specifically we define $f(\bx)=\sum_{j=1}^4\pi_j\no_2(\bx\mid\bgamma_j,\bI)$, with $\bpi=(0.1,0.2, 0.4,0.3)$ and $\bgamma_1'=(1.5,1.5)$, $\bgamma_2'=(-1,1)$, $\bgamma_3'=(-1,-2)$, $\bgamma_4'=(1.5,-1.5)$, and project it to the unit circle. From this model we took two samples of sizes, $n=50$ and $n=500$. For these datasets we fitted our projected P\'olya tree model $\PPT_M(\alpha,\rho,f_0)$. To define the prior we took $f_0=\no_2(\bmu,\bI)$ such that the model is centred on the projected normal distribution, as in \eqref{eq:PPTn}. We varied the value of the location of the bivariate normal to see the effect in the posterior estimation. In particular we took $\bmu'\in\{(0,0),(1,1),(2,2)\}$ and a function $\rho(m)=m^\delta$, with $\delta=1.1$. We also played with different values of the precision parameter $\alpha\in\{0.5,1,2\}$. The depth of the tree was taken as $M=4$. 

We ran our MCMC for $10,000$ iterations with a burn-in of $1,000$ and keeping one of every 5$^{th}$ iteration after burn-in to produce posterior inference. For each prior scenario we computed the LPML statistic. These numbers are summarised in Table \ref{tab:sim1}. Additionally, in Figure \ref{fig:postsim1} we present posterior estimates for most prior scenarios and for $n=500$. The solid line corresponds to the point estimate and the dotted lines form a 95\% credible interval (CI). We accompany all graphs with a probability histogram of the data in the background. 

From Table \ref{tab:sim1} we can see that the fitting becomes worse (smaller LPML values) when the mean of the bivariate normal goes away from the origin. This behaviour was foreseen since the concentration (dispersion) of the projected P\'olya tree prior highly increases (reduces) for larger $||\bmu||$ (see Figure \ref{fig:priormoments}), being harder for the model to adjust to the data. Depending on the value of $\bmu$, some values of $\alpha$ provide better fitting than others. This latter parameter is usually interpreted as a precision parameter in P\'olya trees \citep{hanson:06}. For smaller values of $\alpha$ the model becomes more nonparametric, and more parametric for larger values. Moreover, $\alpha$ also plays the role of a smoothing parameter. This smoothing effect can be appreciated in the top row in Figure \ref{fig:postsim1}, but not so much in the lower rows. What is interesting is that when $\bmu'=(2,2)$ the posterior estimate is highly dependent on the prior and barely moves with the data, despite the large dataset of size $n=500$. 

The best fitting, according to the LPML, is obtained when $\bmu'=(0,0)$ and $\alpha=2$. This is regardless of the data size $n$. The fitting for $n=500$ is depicted in the top-right graph in Figure \ref{fig:postsim1}. The posterior estimate follows smoothly the path of the data. In our experience we do not advise to go beyond $\alpha=2$ unless the data size is very large. All scenarios in Table \ref{tab:sim1} were re-ran with a deeper P\'olya tree with $M=6$, no real advantage was observed in terms of the LPML statistic, but the running time was a lot larger. Perhaps for the case when $\bmu'=(2,2)$, posterior estimates are somehow better than with $M=4$, but still a lot worse than with $\bmu'=(0,0)$. \cite{hanson:06} also obtains a similar conclusion, showing that the LPML stabilises and does not improve for larger $M$, justifying so the use of finite trees. 

Finally, we compare our results with alternative models, specifically we consider a parametric projected normal \citep{nunez&gutierrez:05} and a nonparametric DPM of projected normals \citep{nunez&al:15}. The LPML statistics for these two models, included at the bottom of Table \ref{tab:sim1}, show worse fitting than our PPT. Additionally, comparing the running times for the dataset of size $n=500$, the DPM of projected normals took $1.5$ hours, whereas the PPT took only $10$ minutes for the same amount of iterations.

\subsection{Real data analysis}

In this section, we apply our methodology to the analysis of a real dataset. A study of the interaction among species was carried out as part of a larger research project at El Triunfo biosphere reserve in Mexico in 2015. The use of camera-trapping strategies allowed ecologists to generate temporal activity information (time of the day) for three animal species, peccary, tapir and deer. The data sizes were 16, 35 and 115, respectively, and are  reported in Table \ref{tab:data}. 
 
This data sets has been previously analysed by \cite{nunez&al:18}, using DPM of projected normals to estimate overlapping coefficient among species. Here we fitted our projected P\'olya tree model for each of the directions of the three animals. We centred our prior on a spherical bivariate normal with $\bmu'=(0,0)$, which produces a very dispersed prior projected tree. The concentration function was $\rho(m)=m^{1.1}$ and the depth of the tree was $M=4$. We tried different values of the precision/smoothing parameter $\alpha\in\{0.5,1,2\}$ to compare. The MCMC specifications were the same as for the simulated data, and for each value of $\alpha$ we computed the statistic LPML. 

The goodness of fit statistics are reported in Table \ref{tab:real}. Interestingly, for the tapir and deer datasets the best fitting is achieved with $\alpha=2$, whereas for the peccary dataset the largest LPML value is obtained with $\alpha=0.5$. This is explained by the small data size of peccary which only has $n=16$ points. Alternatively, instead of selecting the best value for $\alpha$ from a range of values, we could place a hyper-prior distribution, say $\alpha\sim\ga(c_\alpha,d_\alpha)$, and update it with its corresponding conditional posterior distribution, which has the form $$f(\alpha\mid\bY)\propto\left\{\prod_{m=0}^{M-1}\prod_{j=1}^{2^m}\prod_{k=1}^{2^m}
\dir(\by_{m,j,k}\mid\balpha_{m,j,k})\right\}\ga(\alpha\mid c_\alpha,d_\alpha),$$ with $\balpha_{m,j,k}=(\alpha\rho(m+1),\ldots,\alpha\rho(m+1))$. Obtaining draws from this distribution requires a MH step. When taking $c_\alpha=1$ and $d_\alpha=2$ we obtain a-priori that $\P(\alpha<2)=0.98$, so we mainly concentrate on values smaller than 2. This produces posterior 95\% CI for $\alpha$, for the three datasets: $(0.17,1.49)$ for peccary, $(0.40,3.00)$ for tapir, and $(0.45,2.61)$ for deer, which are consistent with the selected best values of $\alpha$. The corresponding LPML values are also reported in Table \ref{tab:real}.

Posterior density estimates with the best fitting settings are shown in Figure \ref{fig:postreal} (first column). Point estimates correspond to the solid lines and 95\% CIs to the dotted lines. In all cases, density estimates are multimodal, perhaps for the tapir data the first mode is not so clear. For the peccaries, the directions where they move have a bounded support, mainly from 1.5 to 4.7 radians with a somehow uniform pattern. This range corresponds, approximately, to the time of the day from 6:00 to 18:00 in a 24-hours clock. Looking deeper into the density estimate, we appreciate a bimodal behaviour with peaks at 10:00 and 17:00 hours. On the other hand, tapirs and deer appear everywhere. The mode direction where both tapirs and deer are seen is around 18:00 hrs. ($3\pi/2$ radians).

We use the mean $\nu_\theta$, as in \eqref{eq:moments}, to summarise the preferred direction. We have assumed that the density $f$ of the directions $\theta$ is nonparametric, therefore the mean of $\theta$ is not a singe value, but a set of values whose probability distribution can be obtained. Posterior distribution for the mean direction of the three animals are presented in Figure \ref{fig:postmeans} as boxplots. On a 24-hours clock, the mean direction for peccaries goes from 10:02 to 14:49 hours ($2.63$ to $3.88$ radians) with 95\% probability, the mean direction for tapirs goes from 18:11 to 23:09 hours ($-1.52$ to $-0.22$ radians) with 95\% probability, and finally, the mean direction for deer goes from 16:37 to 21:35 hours ($-1.93$ to $-0.63$ radians) with 95\% probability. Broadly speaking we can say that peccaries have a preferred activity-time around midday, which is totally different to the other two animals whose preferred times are around 20:30 hours, for tapirs, and around 19:00 hours, for deer. We can formally test whether $\nu_\theta^{tapir}=\nu_\theta^{deer}$ by computing a 95\% credible interval of the difference, $\nu_\theta^{tapir}-\nu_\theta^{deer}\in(-0.48,1.31)$ radians, which clearly includes the value of zero. Additionally, $\P(\nu_\theta^{tapir}>\nu_\theta^{deer}\mid\data)=0.8$, which is not big enough to declare a difference. 

Although our projected P\'olya tree model produces smooth densities, further smoothing can be achieved if we also assign a hyper-prior distribution to the location parameter $\bmu$ of the projected normal centering measure $f_0$, inducing a mixture in the nested partitions. If in particular we take $f(\mu_1,\mu_2)=\no(\mu_1\mid\gamma_0,\tau_\mu)\no(\mu_2\mid\gamma_0,\tau_\mu)$, then the conditional posteriors are conjugate and become $f(\mu_1\mid\br)=\no\left(\mu_1\mid\left(\sum_{i=1}^n r_i\cos(\theta_i)+\tau_\mu\gamma_0\right)/(n+\tau_\mu),n+\tau_\mu\right)$ and $f(\mu_2\mid\br)=\no\left(\mu_2\mid\left(\sum_{i=1}^n r_i\sin(\theta_i)+\tau_\mu\gamma_0\right)/(n+\tau_\mu),n+\tau_\mu\right)$. We fitted this mixture of PPT with $\gamma_0=0$, $\tau_\mu=1$, a hyper-prior $\alpha\sim\ga(1,2)$ and the same MCMC specifications as above. The LPML values (ante penultimate row in Table \ref{tab:real}) are not as good as those without mixing, however posterior density estimates turn out to be a lot smoother (see second column in Figure \ref{fig:postreal}). Posterior means of the location parameters are $\widehat\bmu'=(-0.81,-0.07)$ for peccary, $\widehat\bmu'=(0.28,-0.59)$ for tapir and $\widehat\bmu'=(0.08,-0.34)$ for deer. 

Finally, we also compare with the parametric projected normal and the nonparametric DPM of projected normals. The corresponding LPML statistics are reported in the last two rows of Table \ref{tab:real}. Now, comparing the best fit of our PPT model with the two competitors we have very interesting findings. For the peccary dataset our proposal is by far the best model with the DPM in second and the parametric model in third place. For the tapir dataset the three models practically achieve the same fit. In an attempt to explain why this happens, we recall that circular data have no beginning and end points, so histograms are better seen in a circle. The first block of points, close to zero, can be seen as a continuation of the larger points, close to $2\pi$, and therefore we could appreciate a single predominant mode characterising the data, thus a parametric (unimodal) model, like the projected normal, could do a good job describing this dataset. Lastly, for the deer data, the best fit is obtained by the DPM, followed closely by our PPT and the parametric model in a far third place. 

As suggested by one of the referees, it is straightforward to compute a Bayes factor for testing the adequacy of the underlying projected normal model by using the Savage-Dickey ratio \citep{dickey:71,hanson:06}. For the hypothesis $H_0:\theta_i\sim{\rm Proj.Normal}\iff Y_{m,j,k}=1/4\;\forall(m,j,k)$ versus $H_1:\theta_i\sim\PPT\iff Y_{m,j,k}\neq 1/4 \mbox{ for some }(m,j,k)$ the Bayes factor would be $BF_{10}=f_{\bY}(1/4)/f_{\bY\mid\data}(1/4)$, that is the ratio of the prior and posterior densities of $\bY$ evaluated at the null hypothesis. For both peccary and tapir datasets we obtain $BF_{10}=1.61$ which barely favours the PPT alternative, whereas for deer dataset $BF_{10}=0.20$ slightly supporting the parametric model.

\section{Concluding remarks}
\label{sec:concl}

We have proposed a Bayesian nonparametric model for circular data. Our proposal is based on the projection of a bivariate P\'olya tree to the unit circle. Random densities obtained from the model turned out to be smooth. This is in contrast to the bivariate densities obtained from a bivariate P\'olya tree which are discontinuous at the boundaries of the partitions. 

Posterior inference is simply done by augmenting the data with (unobserved) latent resultants and updating the bivariate tree. To simplify the posterior dependence on the prior choice of $\alpha$, we suggest to place a hyper-prior on this parameter, with minimal extra effort in sampling from its conditional posterior distribution. Extra smoothing on the density estimates can be achieved by placing a hyper prior on the parameter $\bmu$ producing a mixture of PPT. Comparing the performance of our model to other alternatives, for the three datasets studied here, we showed that our proposal is a good competitor, with the advantage of the simplicity in the posterior inference. 

Generalising our model to directional data with more than two dimensions would require to project a multivariate P\'olya tree on $\mathbb{R}^k$ to the unit sphere $\mathbb{S}^k$. This can be done straightforwardly by generalising the nested partition to include sets $\Pi=\{B_{m,j_1,\ldots,j_k}\}$, where at each level $m$ the partition size would be $2^{km}$. Studying the properties of this generalisation remains open and is left for future work. 

Additionally, the inclusion of covariates in the projected (bivariate) P\'olya tree also deserves study. Specifically, if $\bX_i'=(X_{1i},X_{2i})=(R_i\cos\Theta_i,R_i\sin\Theta_i)$, where $\Theta_i$ are observable angles and $R_i$ are latent variables, and if $\bZ_i$ is a $p$-vector of covariates, then a regression model would be $\bX_i=B\bZ_i+\bepsilon_i$, where $B$ is a $(2\times p)$-matrix of coefficients and $\bepsilon_i'=(\epsilon_{1i},\epsilon_{2i})\mid f$ are i.i.d. and $f\sim\PT_M(\alpha,\rho,F_0)$ is a finite bivariate P\'olya tree as in \eqref{eq:dPT}. The induced distribution of $\Theta_i\mid\bZ_i$ would be a projected P\'olya tree regression model.

\section*{Acknowledgements}

This work was partial supported by the National System of Researchers, Mexico. The first author acknowledges support from \textit{Asociaci\'on Mexicana de Cultura, A.C.} Finally, the authors are deeply thankful to professor Eduardo Mendoza from \textit{Universidad Michoacana de San Nicol\'as de Hidalgo}, Mexico for the corresponding permission to use the data from project at El Triunfo biosphere reserve in Mexico.

\newpage

\begin{table}
\caption{LPML goodness of fit measures for simulated data. First block corresponds to our PPT.}
\label{tab:sim1}
\begin{center}
\begin{tabular}{cccc} \hline \hline
 	& 	&	\multicolumn{2}{c}{LPML} \\ 
$\bmu'$	&	$\alpha$	&	$n=50$	&	$n=500$ \\ \hline
$(0,0)$ & $0.5$ & $-88.10$ & $-849.73$ \\
$(0,0)$ & $1.0$ & $-87.07$ & $-848.25$ \\
$(0,0)$ & $2.0$ & $\bf -86.98$ & $\bf -847.70$ \\
$(1,1)$ & $0.5$ & $-89.13$ & $-848.04$ \\
$(1,1)$ & $1.0$ & $-91.41$ & $-848.07$ \\
$(1,1)$ & $2.0$ & $-95.82$ & $-851.15$ \\
$(2,2)$ & $0.5$ & $-118.06$ & $-1050.26$ \\
$(2,2)$ & $1.0$ & $-129.12$ & $-1065.09$ \\
$(2,2)$ & $2.0$ & $-147.12$ & $-1093.84$ \\ \hline
\multicolumn{2}{c}{Proj.Normal}  & $-87.54$ & $-864.90$ \\
\multicolumn{2}{c}{DPM Proj.Normal} &  $-87.01$ & $-848.29$ \\
\hline \hline
\end{tabular}
\end{center}
\end{table}

\begin{table}
\caption{Temporal activity (in radians) from camera trap records relating to the presence of three mammalian species at El Triunfo reserve.}
\label{tab:data}
\centering
\begin{tabular}{llllllll} \hline \hline
\multicolumn{8}{c}{Peccary} \\ \hline
3.0757 & 2.7422 & 3.2214 & 0.8017 & 2.3065 & 2.6849 & 4.5517 & 4.3300 \\ 2.3421 & 4.6541 & 2.2754 & 2.4580 & 3.3150 & 4.0887 & 4.4092 & 4.2632 \\  
\hline	\multicolumn{8}{c}{Tapir} \\ \hline
3.3352 & 4.6813 & 4.7835 & 5.4591 & 5.4929 & 3.6559 & 4.9567 & 4.5505 \\ 3.7114 & 4.6214 & 5.5011 & 0.7815 & 0.4264 & 5.6929 & 4.6098 & 0.0712 \\ 4.7340 & 4.7583 & 0.8511 & 4.5465 & 4.0871 & 1.3747 & 4.8558 & 0.9962 \\ 4.9629 & 2.7328 & 5.9844 & 0.6099 & 5.9213 & 1.9393 & 6.2521 & 4.7322 \\ 4.8155 & 5.1034 & 0.5203 \\
\hline	\multicolumn{8}{c}{Deer} \\ \hline
4.5338 & 4.9636 & 2.3963 & 0.1049 & 0.6435 & 1.6665 & 2.7504 & 0.5619 \\ 5.2474 & 4.5670 & 4.4406 & 5.3001 & 4.6440 & 0.8320 & 1.5593 & 2.6858 \\ 5.3614 & 1.5104 & 2.1596 & 4.5811 & 4.9057 & 6.1155 & 1.9216 & 3.6685 \\ 4.7676 & 4.1158 & 3.3225 & 1.0981 & 4.7476 & 2.0472 & 4.0766 & 4.4075 \\ 4.4901 & 5.6538 & 5.4914 & 2.0064 & 5.8532 & 0.0833 & 2.3170 & 0.6101 \\ 5.3250 & 0.7459 & 3.4606 & 4.8188 & 4.4032 & 4.2024 & 1.5408 & 5.3556 \\ 5.2969 & 5.9074 & 5.1198 & 4.7095 & 4.9927 & 1.5943 & 4.8544 & 0.9802 \\ 4.7600 & 4.8139 & 4.9786 & 2.3377 & 5.0841 & 4.1202 & 6.2377 & 2.7648 \\ 4.7023 & 4.3310 & 2.5126 & 6.0751 & 2.2459 & 1.2403 & 2.7941 & 5.0400 \\ 5.3202 & 1.4342 & 3.2619 & 1.9663 & 4.7633 & 5.7232 & 2.1505 & 3.9069 \\ 0.8642 & 3.5219 & 4.9393 & 2.3317 & 4.0359 & 2.0050 & 5.4570 & 4.6069 \\ 6.0874 & 0.1445 & 0.9540 & 3.4935 & 1.6002 & 5.2741 & 0.5729 & 6.1006 \\ 1.0324 & 4.8253 & 5.9624 & 3.5083 & 4.3276 & 4.6632 & 0.6040 & 0.7223 \\ 3.4750 & 5.1140 & 4.9180 & 4.2155 & 4.5710 & 0.5368 & 5.1135 & 3.1823 \\ 3.1831 & 4.4513 & 5.5457 \\
\hline 	\hline 
\end{tabular}
\end{table}

\begin{table}
\caption{LPML goodness of fit measures for El Triunfo Reserve data. First two blocks correspond to our PPT model.}
\label{tab:real}
\begin{center}
\begin{tabular}{cccc} \hline \hline
$\alpha$ & Peccary & Tapir & Deer \\ \hline
$0.5$ & $\bf -23.05$ & $-61.02$ & $-208.31$ \\
$1$ & $-23.22$ & $-60.20$ & $-206.92$ \\
$2$ & $-24.10$ & $\bf -59.57$ & $\bf -205.68$ \\ 
$\ga(1,2)$ & $-23.40$ & $-60.15$ & $-206.77$ \\ \hline
$\alpha\sim\ga(1,2)$, $\mu_j\sim\no(0,1)$ &  $-31.58$ & $-65.49$ & $-212.17$ \\ \hline
Proj.Normal  & $-26.52$ & $-59.43$ & $-207.54$	\\
DPM Proj.Normal &  $-24.64$ & $-59.56$ & $-204.31$ \\ 
\hline \hline
\end{tabular}
\end{center}
\end{table}

\begin{figure}
\centerline{\includegraphics[scale=0.31]{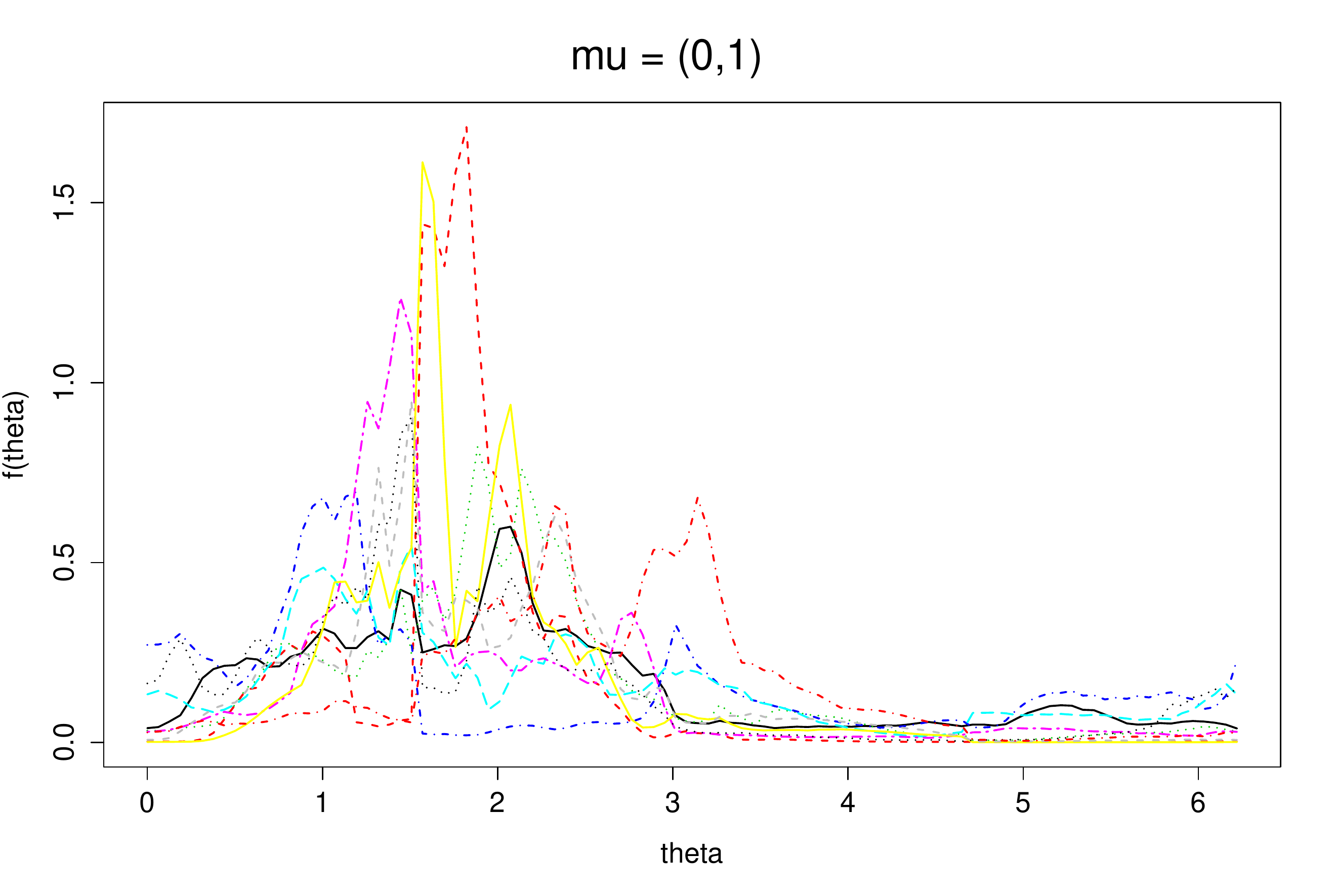}
\includegraphics[scale=0.31]{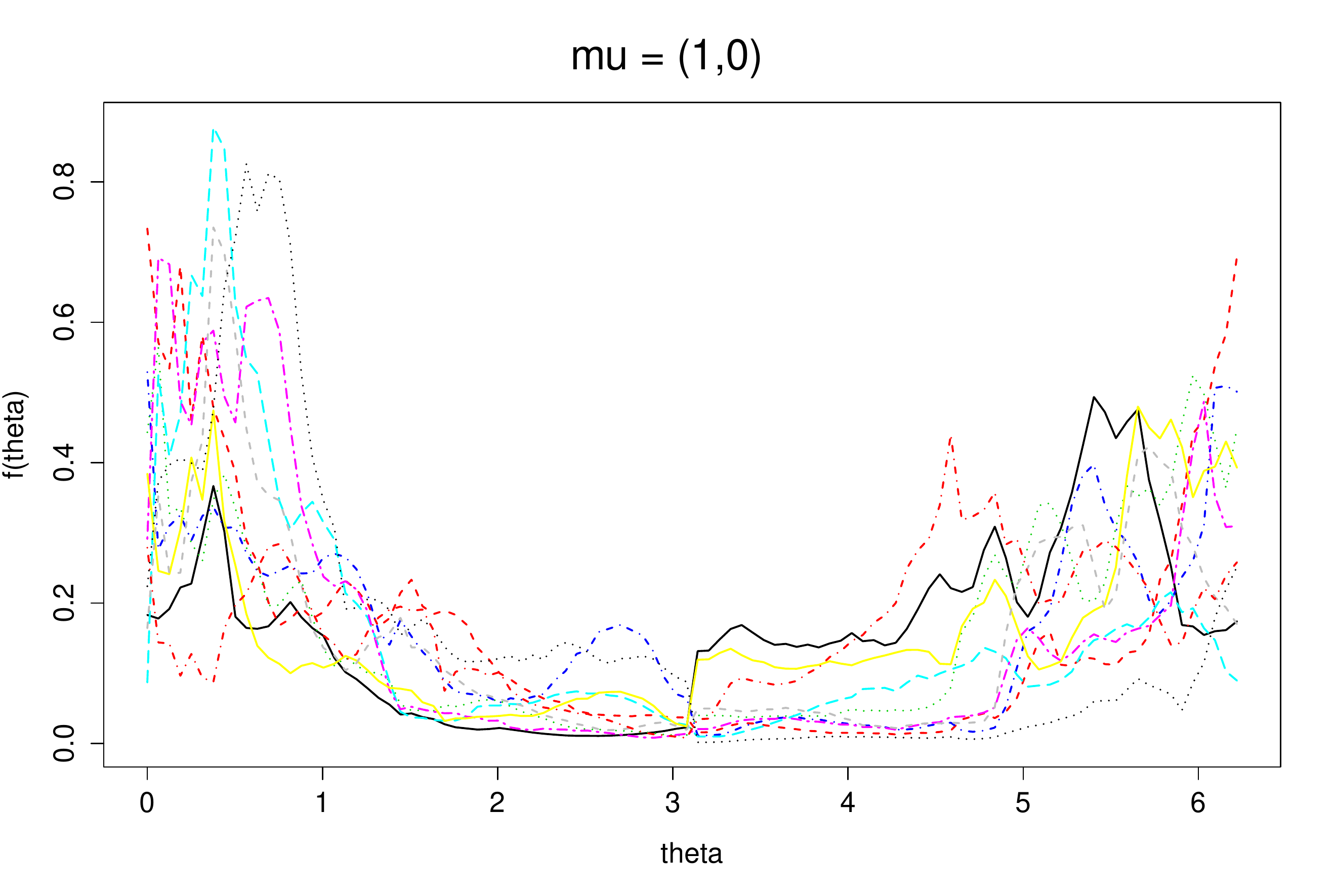}}
\centerline{\includegraphics[scale=0.31]{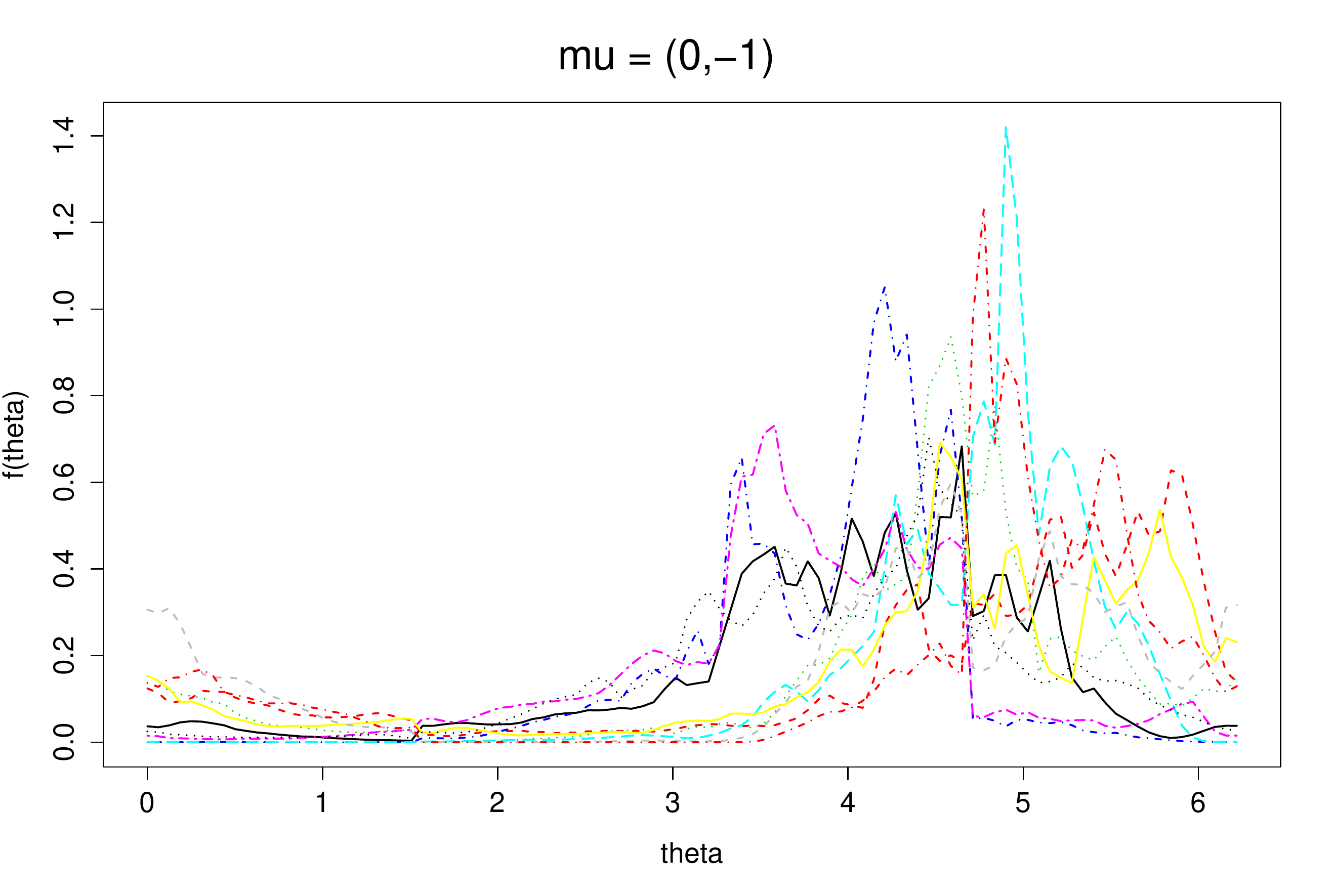}
\includegraphics[scale=0.31]{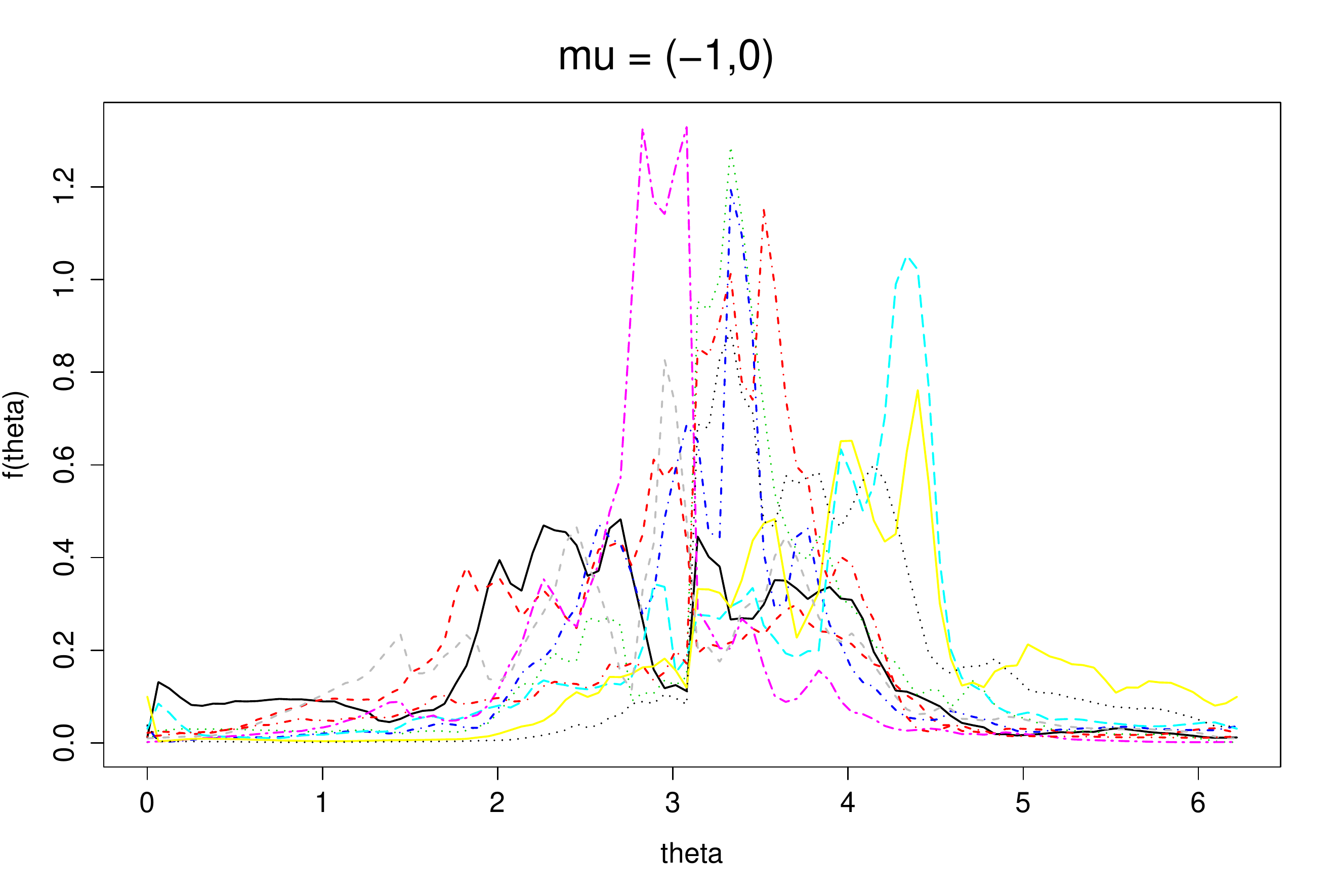}}
\caption{{\small Ten simulated densities of the prior projected P\'olya tree with $M=4$, $\alpha=1$, $\delta=1.1$, for varying $\bmu'$.}}
\label{fig:priorppt1}
\end{figure}

\begin{figure}
\centerline{\includegraphics[scale=0.31]{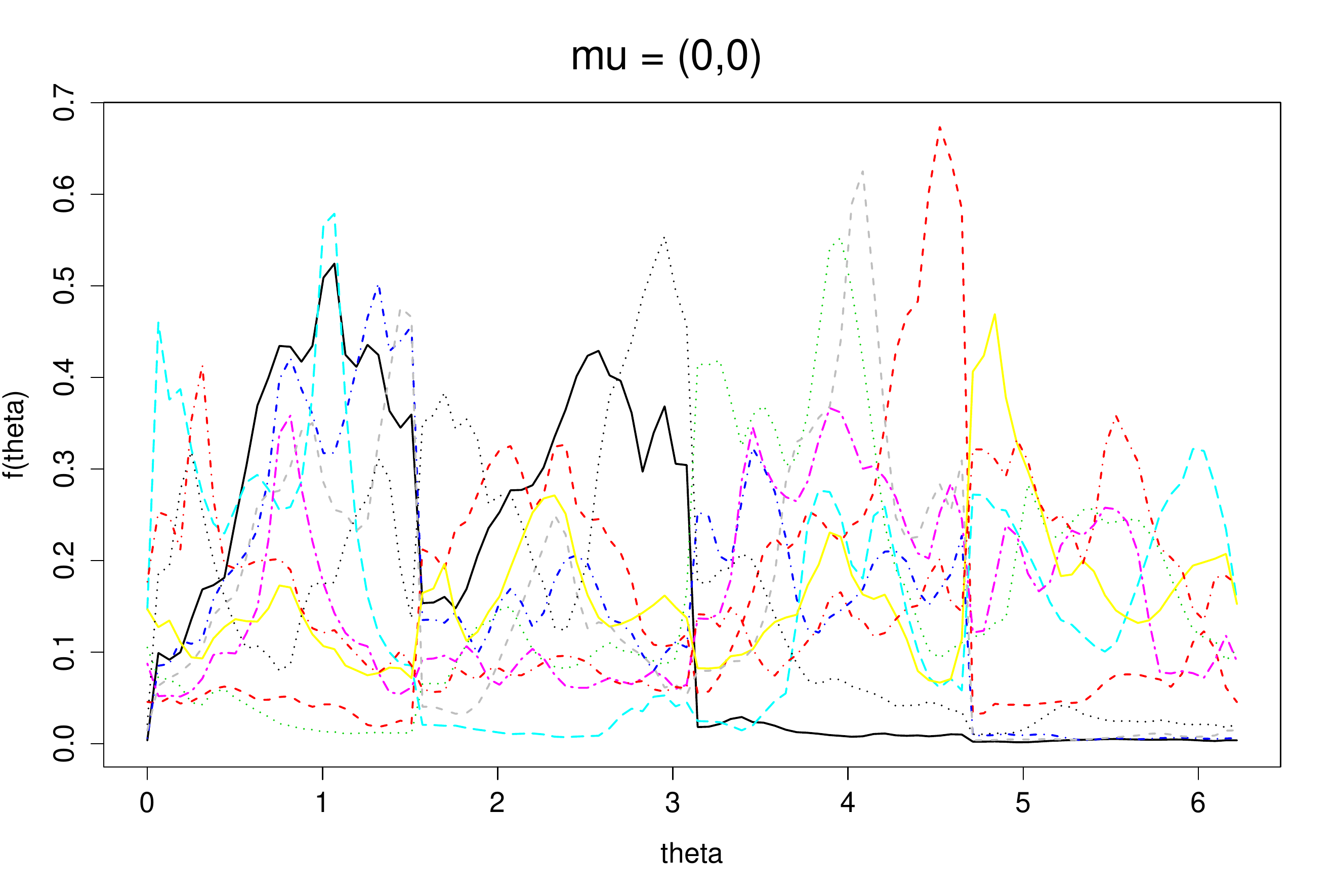}
\includegraphics[scale=0.31]{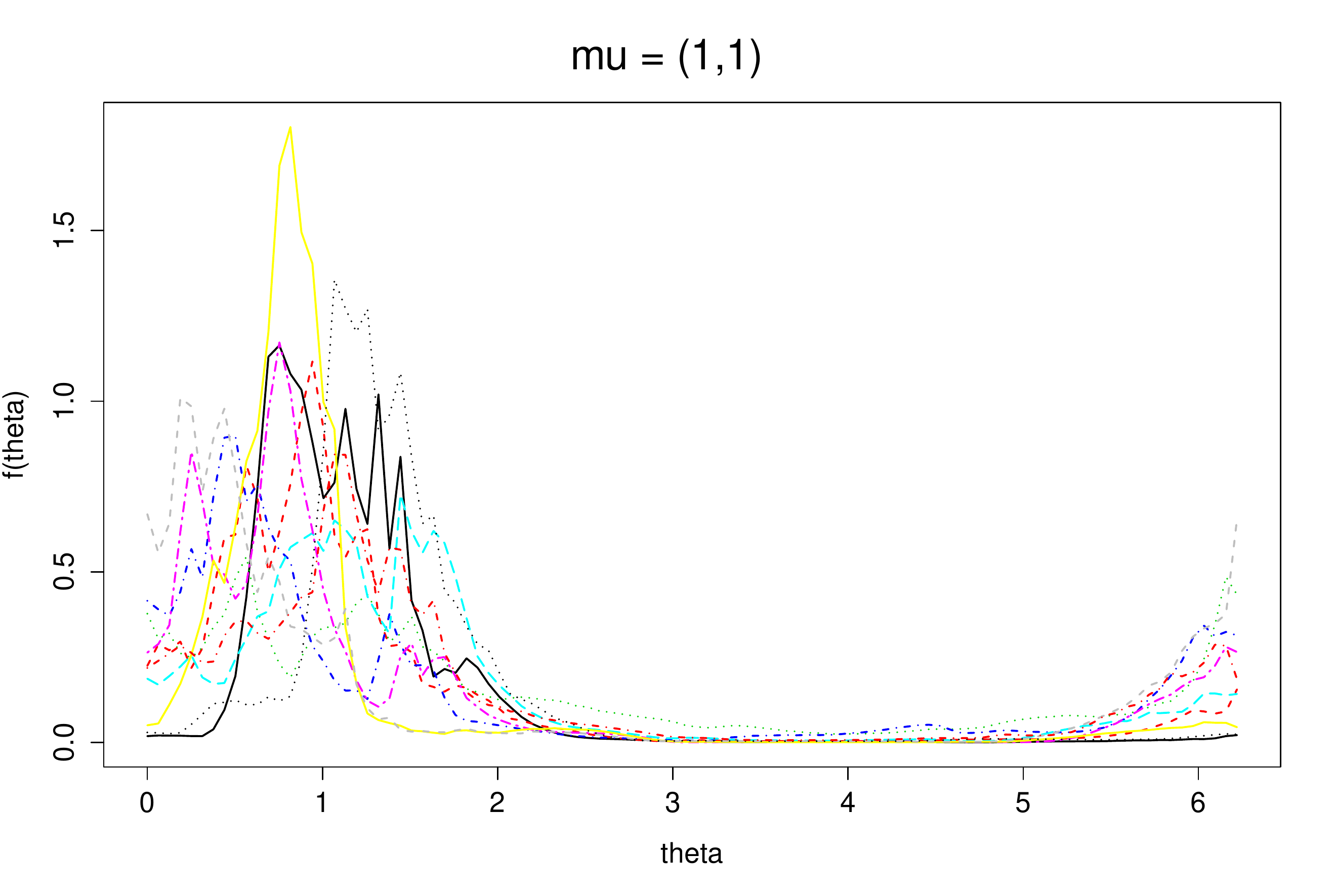}}
\centerline{\includegraphics[scale=0.31]{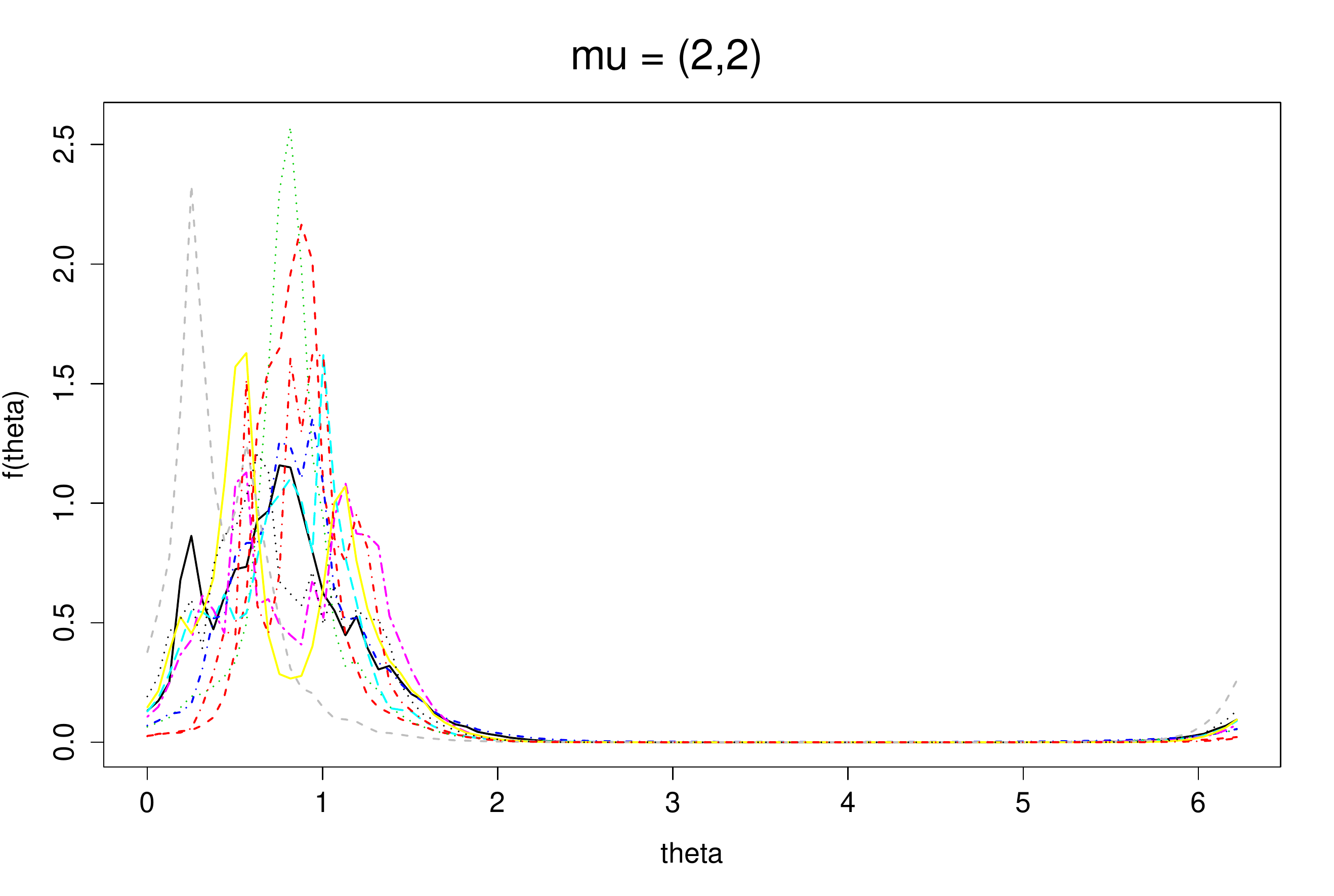}
\includegraphics[scale=0.31]{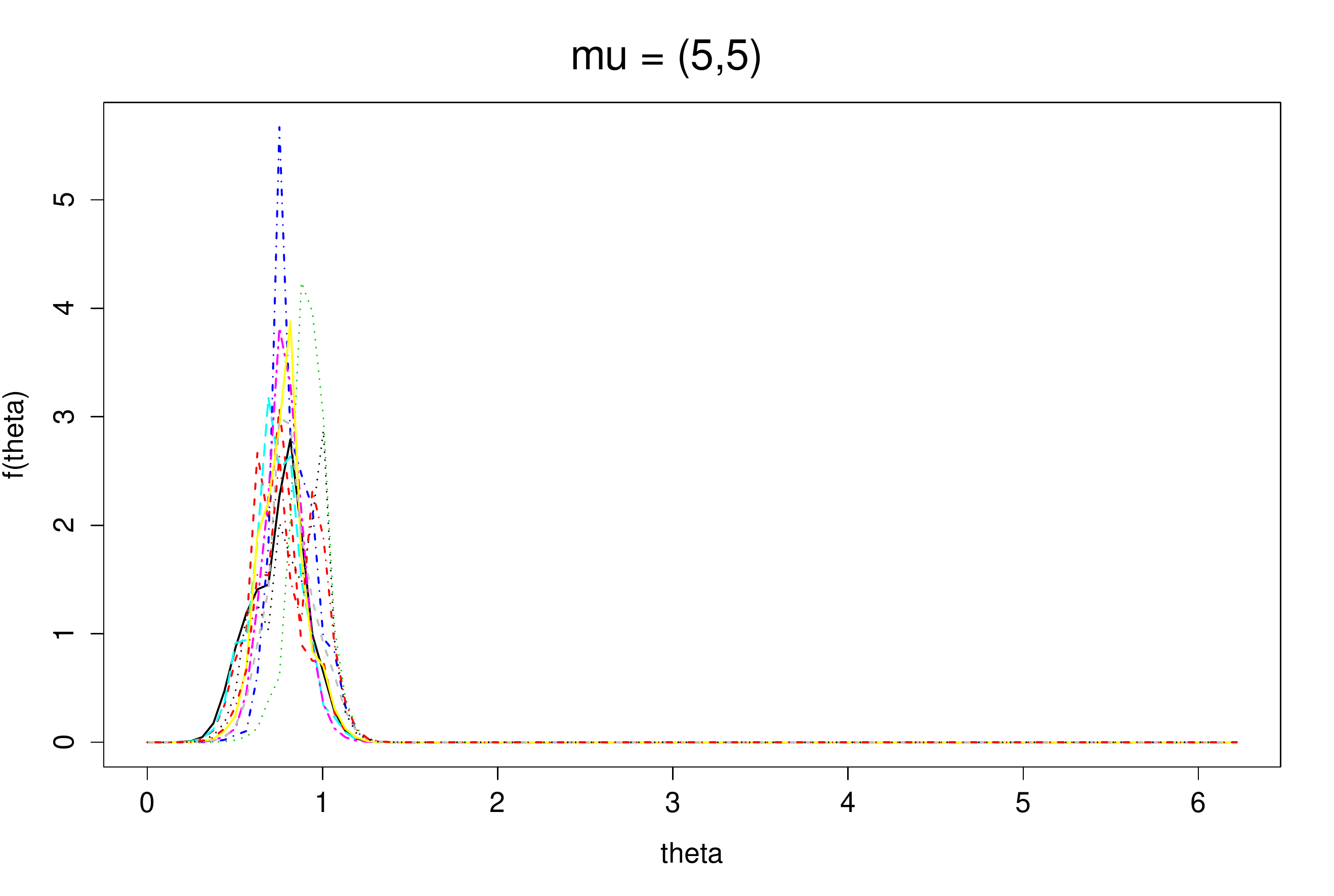}}
\caption{{\small Ten simulated densities of the prior projected P\'olya tree with $M=4$, $\alpha=1$, $\delta=1.1$, for varying $\bmu'$.}}
\label{fig:priorppt2}
\end{figure}

\begin{figure}
\centerline{\hspace{1cm}Mean $\nu_{\theta}$\hspace{6.3cm}Concentration $\varrho_\theta$}
\centerline{\includegraphics[scale=0.32]{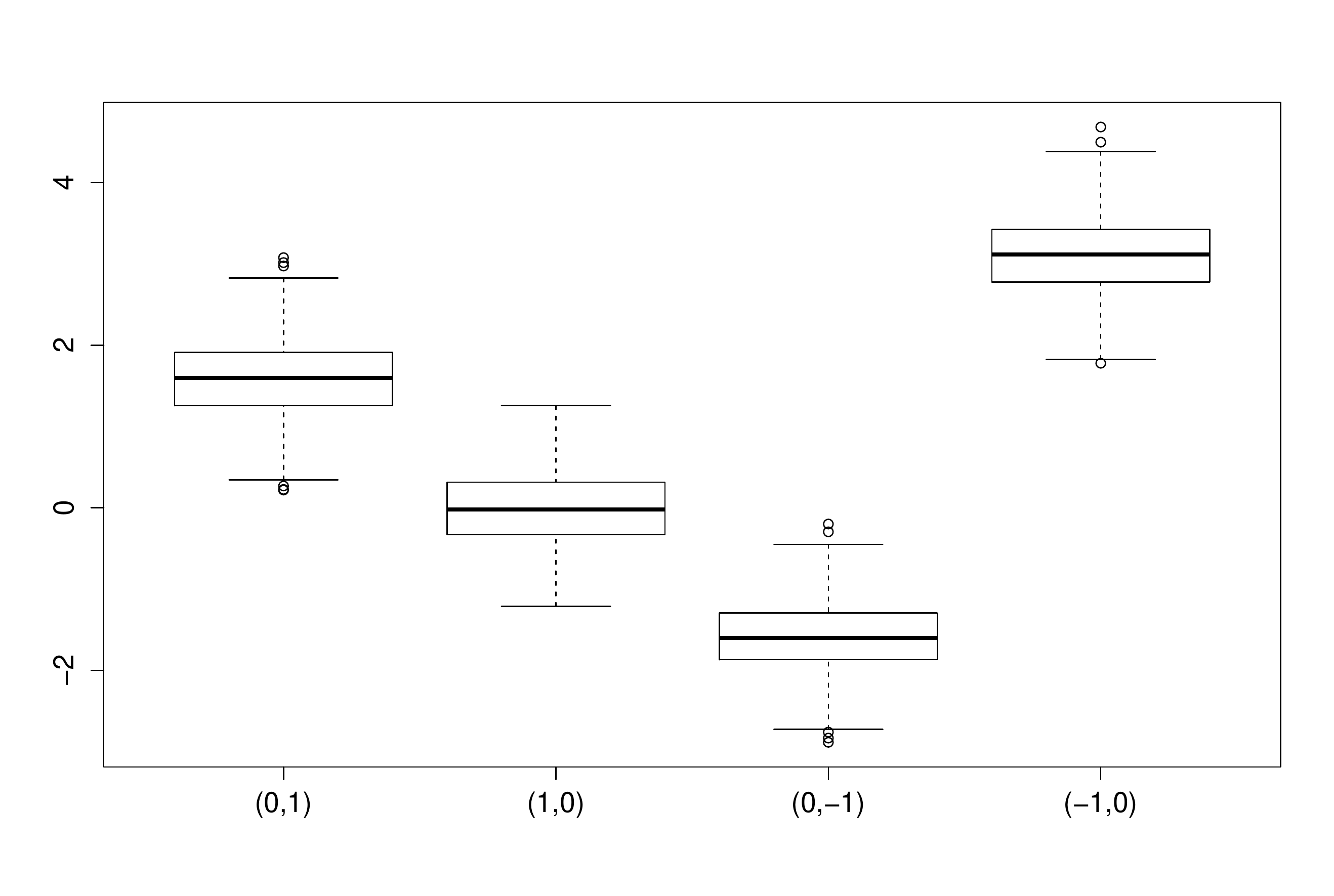}
\includegraphics[scale=0.32]{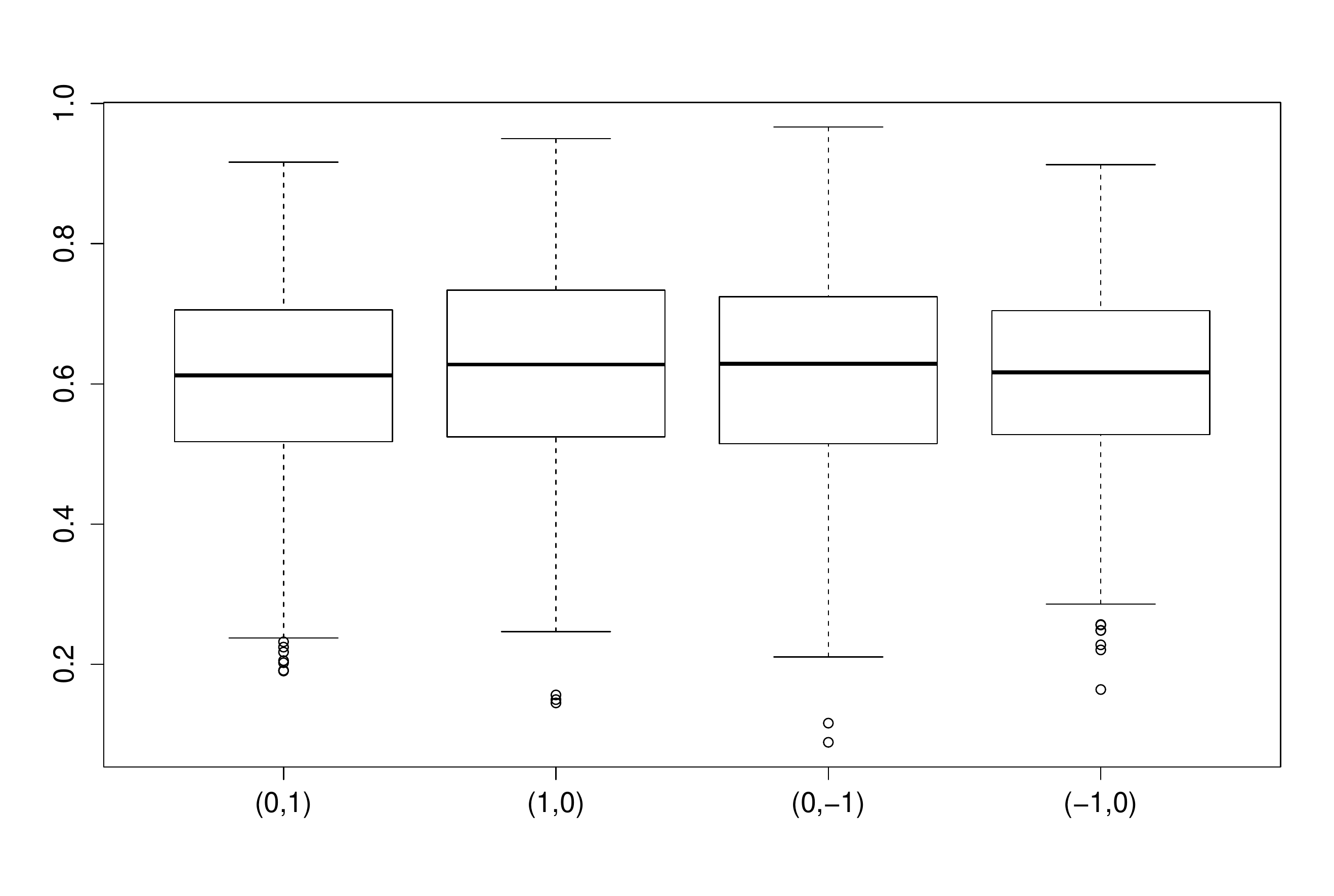}}
\centerline{\includegraphics[scale=0.32]{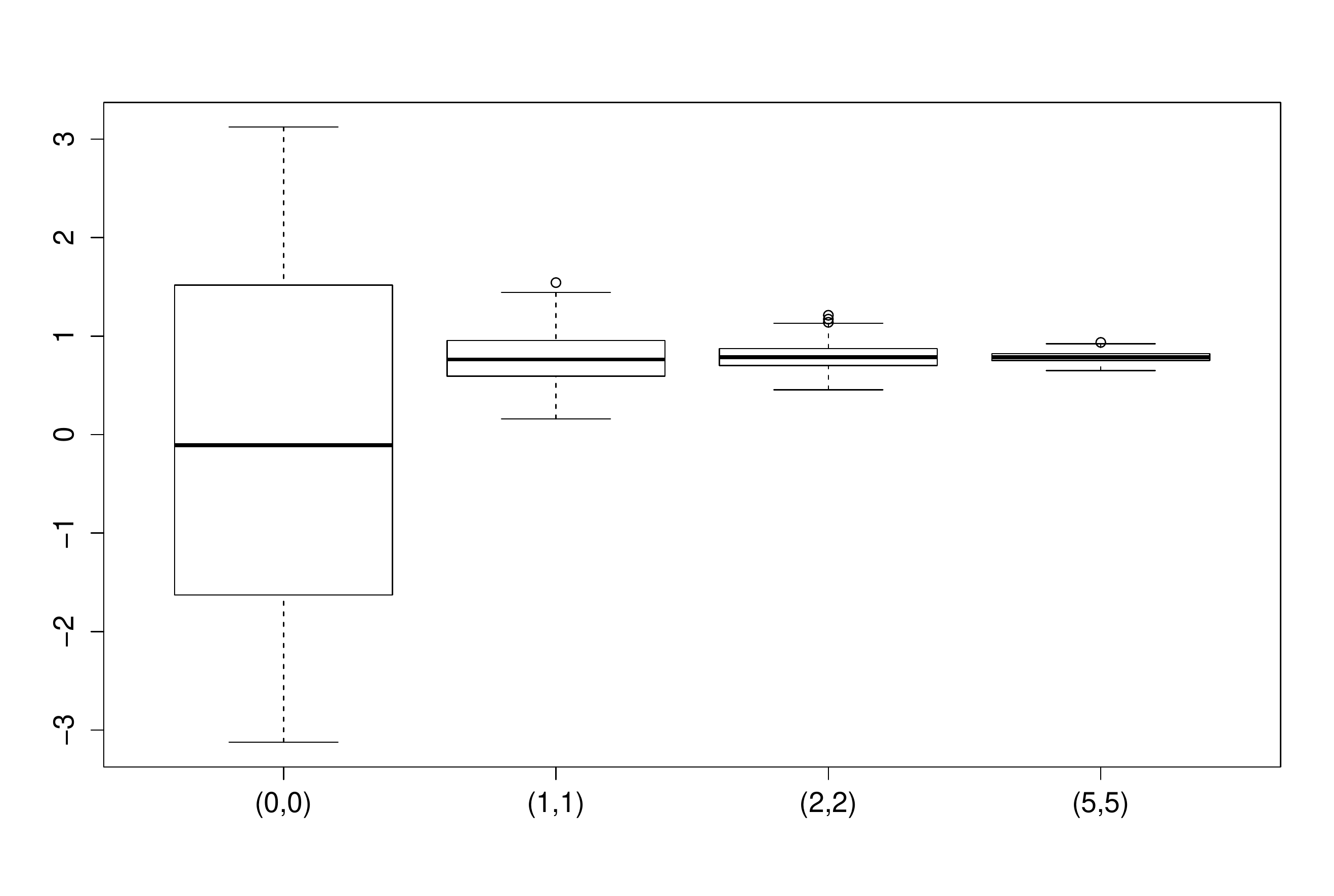}
\includegraphics[scale=0.32]{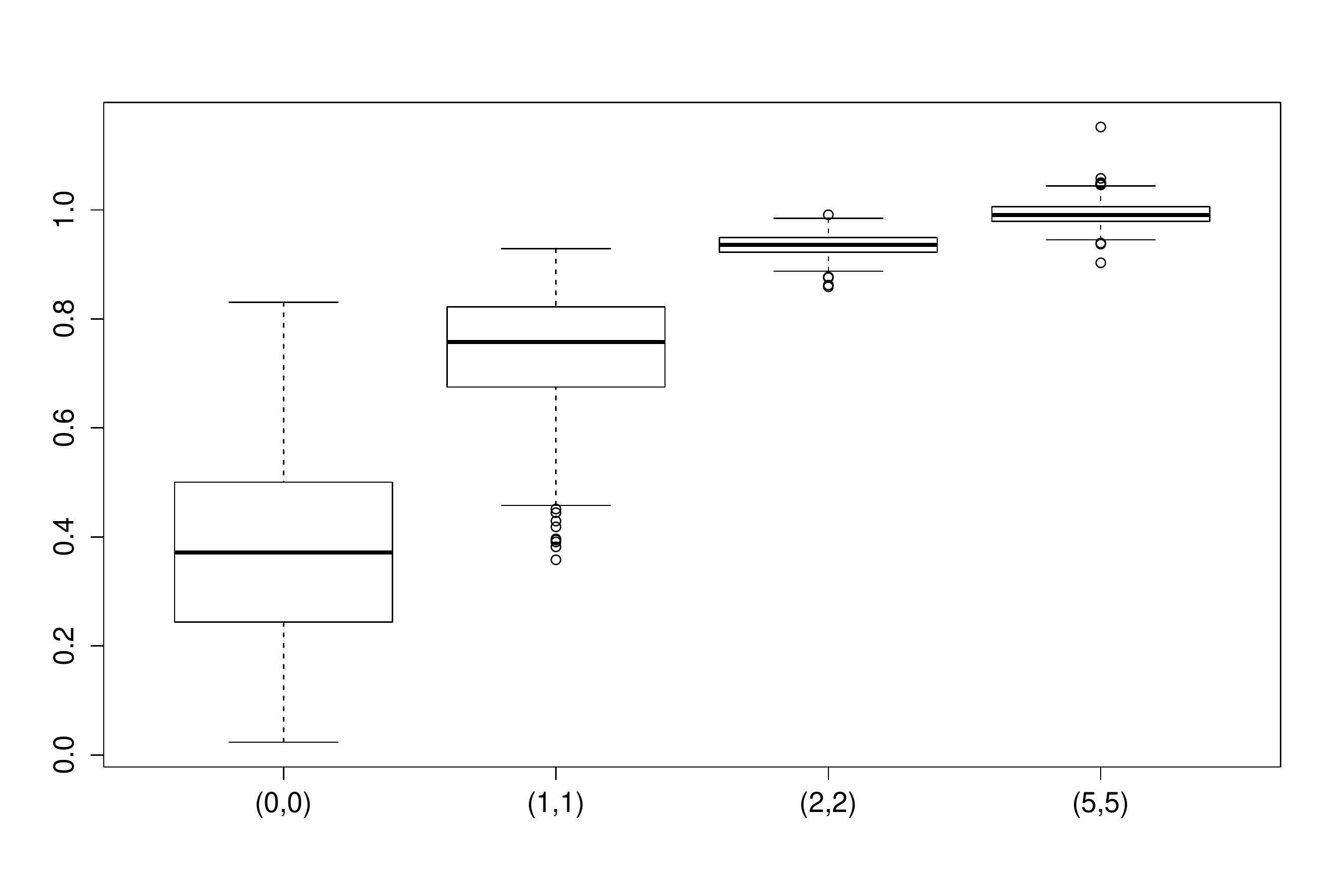}}
\caption{{\small Prior distribution (boxplot) of moments for the projected P\'olya tree with $M=4$, $\alpha=1$, $\delta=1.1$, for varying $\bmu'$.}}
\label{fig:priormoments}
\end{figure}

\begin{figure}
\centerline{\includegraphics[scale=0.31]{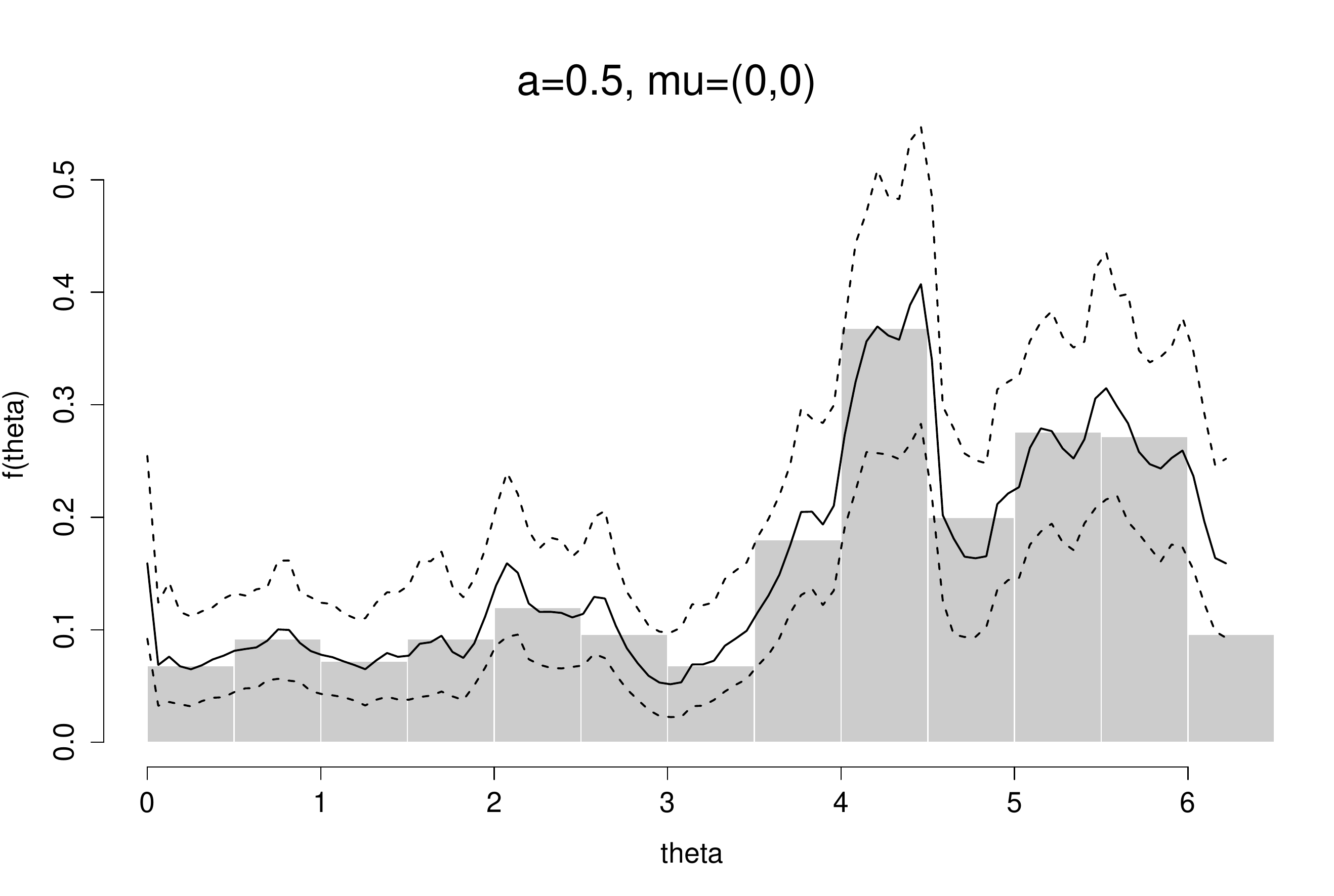}
\includegraphics[scale=0.31]{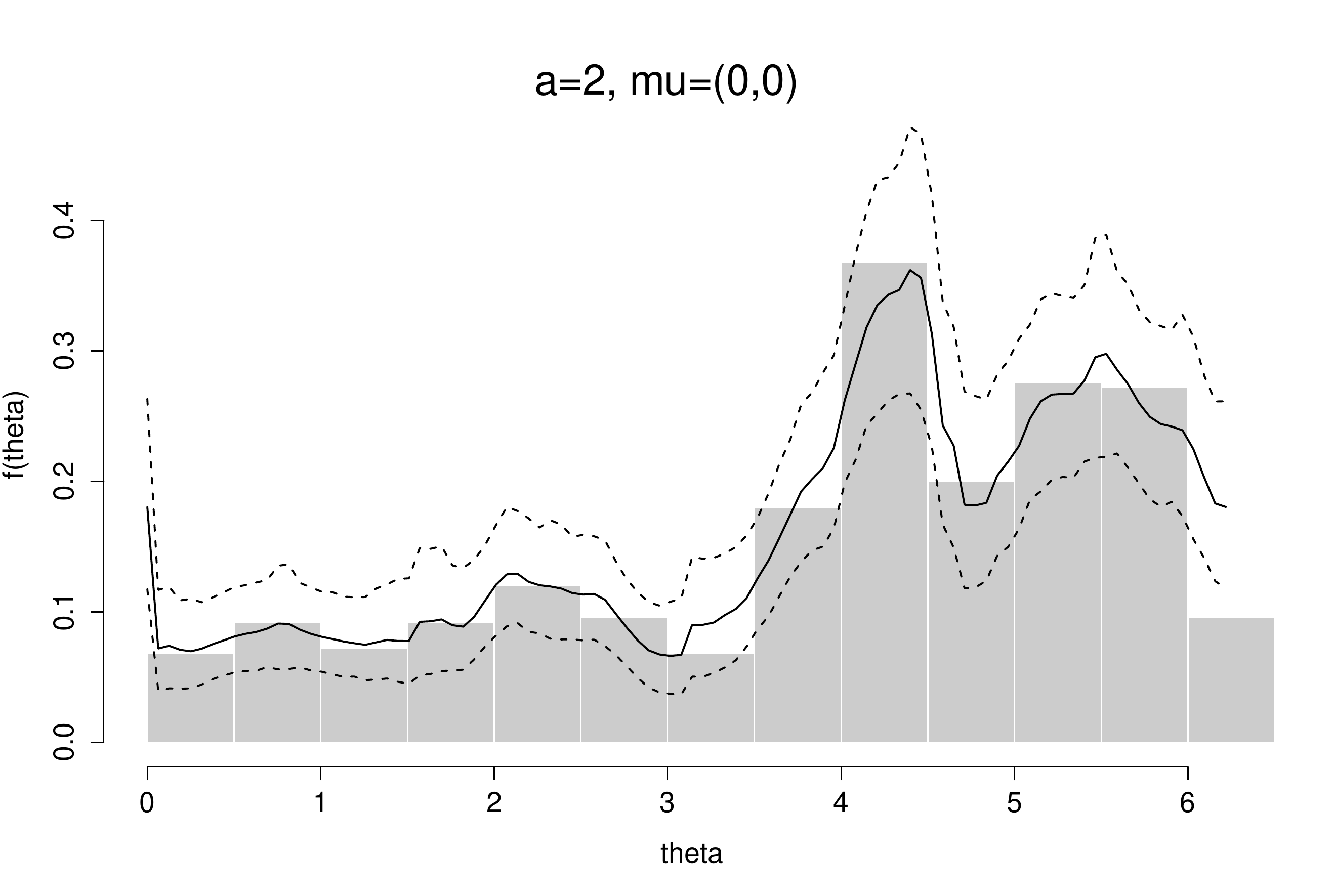}}
\centerline{\includegraphics[scale=0.31]{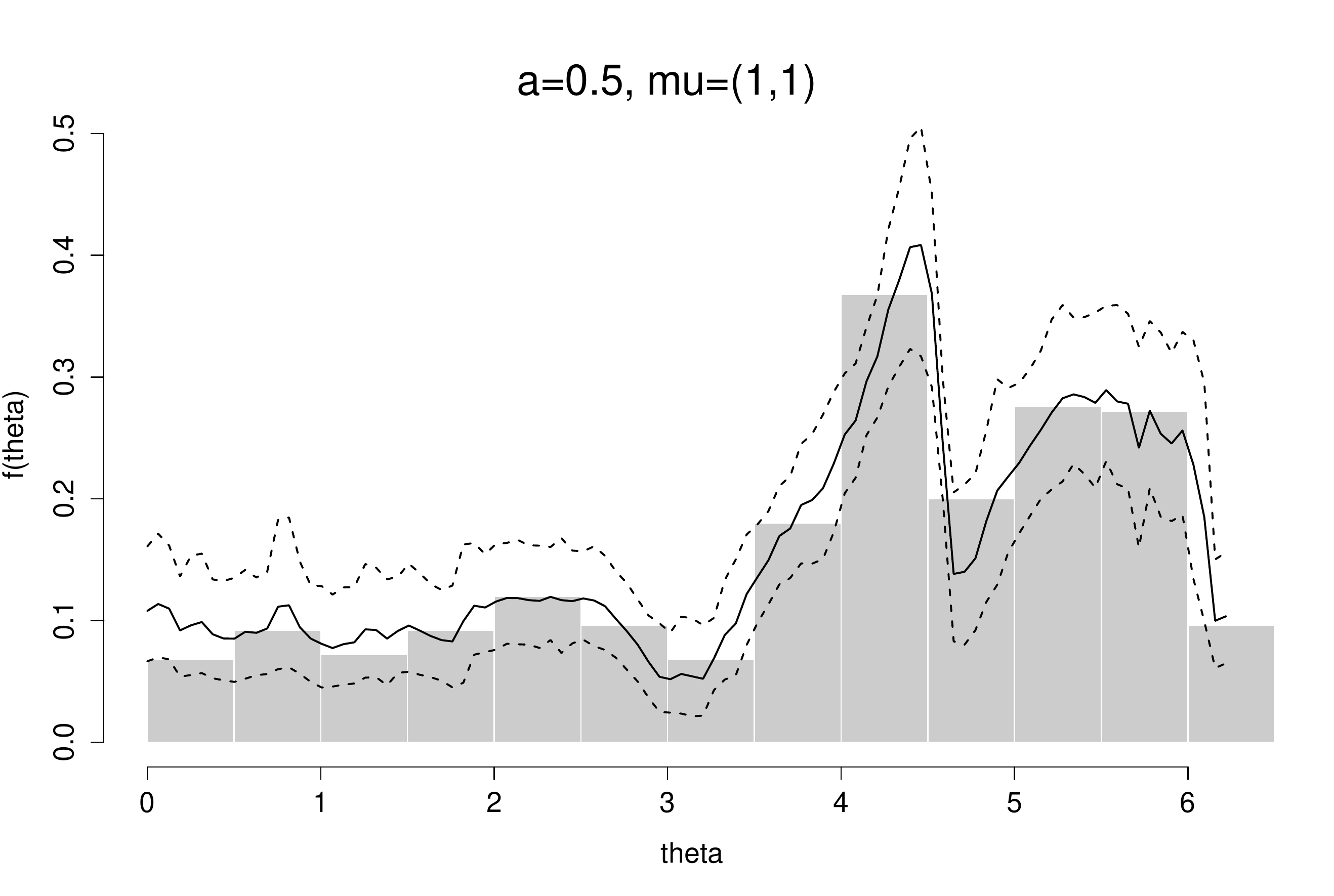}
\includegraphics[scale=0.31]{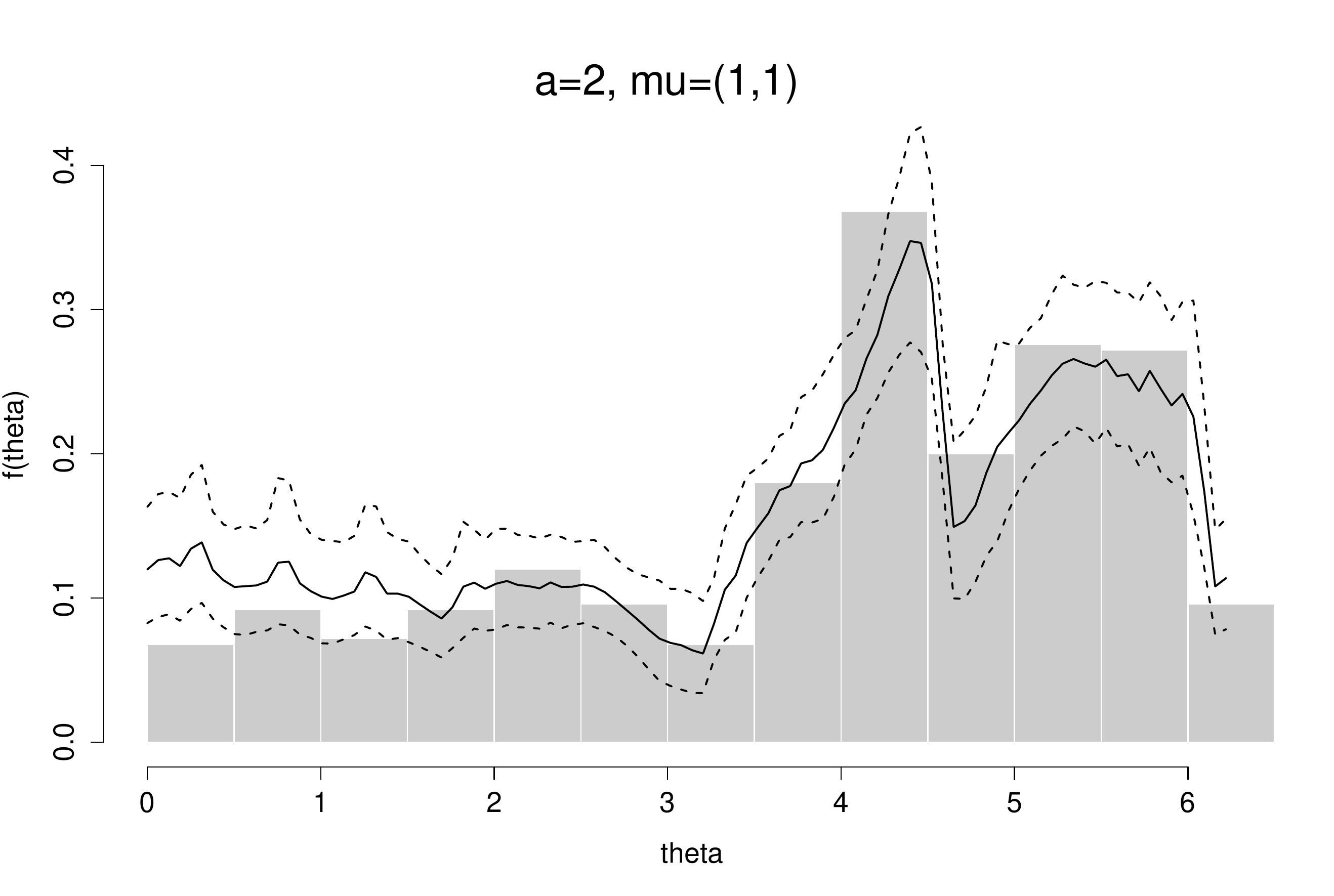}}
\centerline{\includegraphics[scale=0.31]{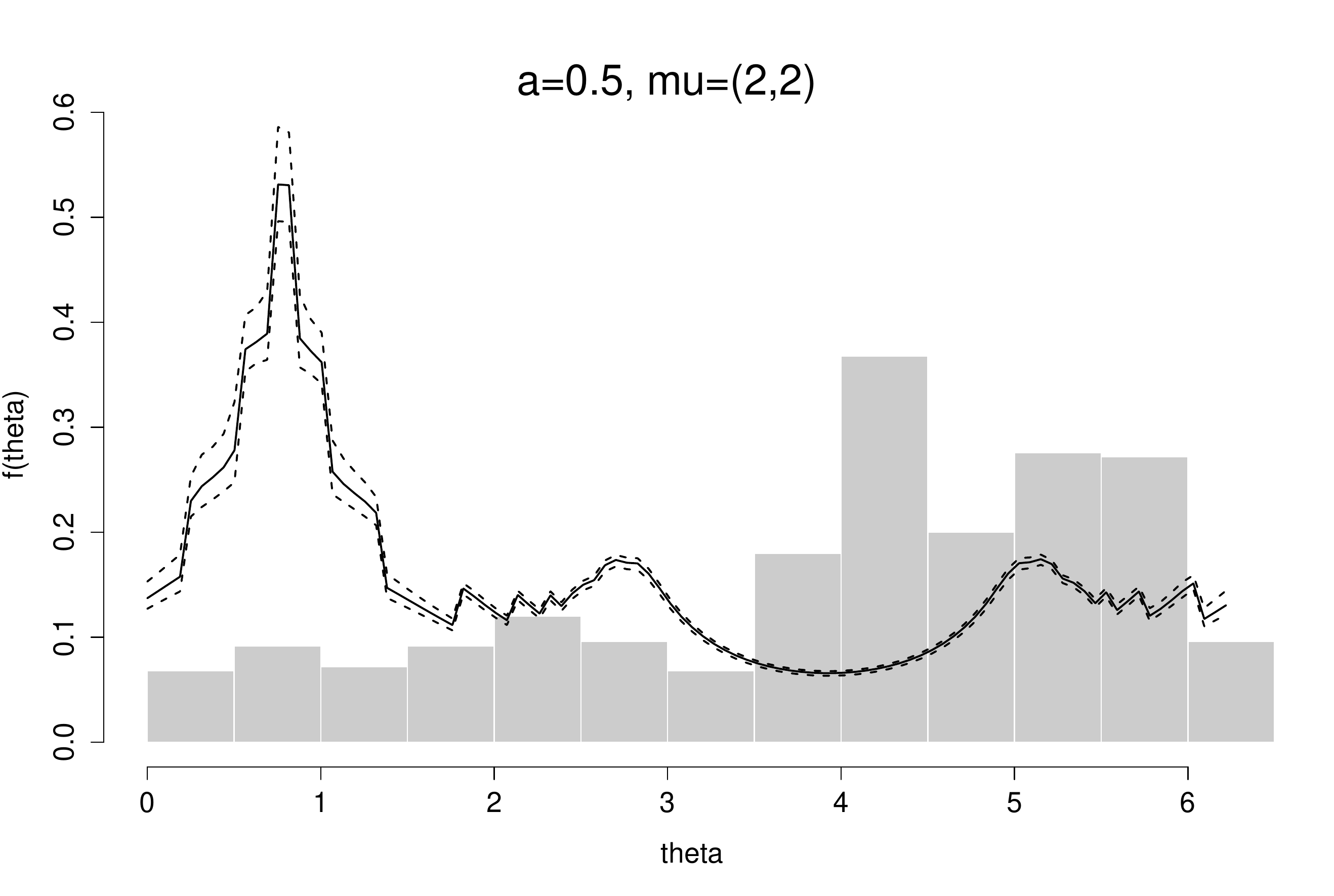}
\includegraphics[scale=0.31]{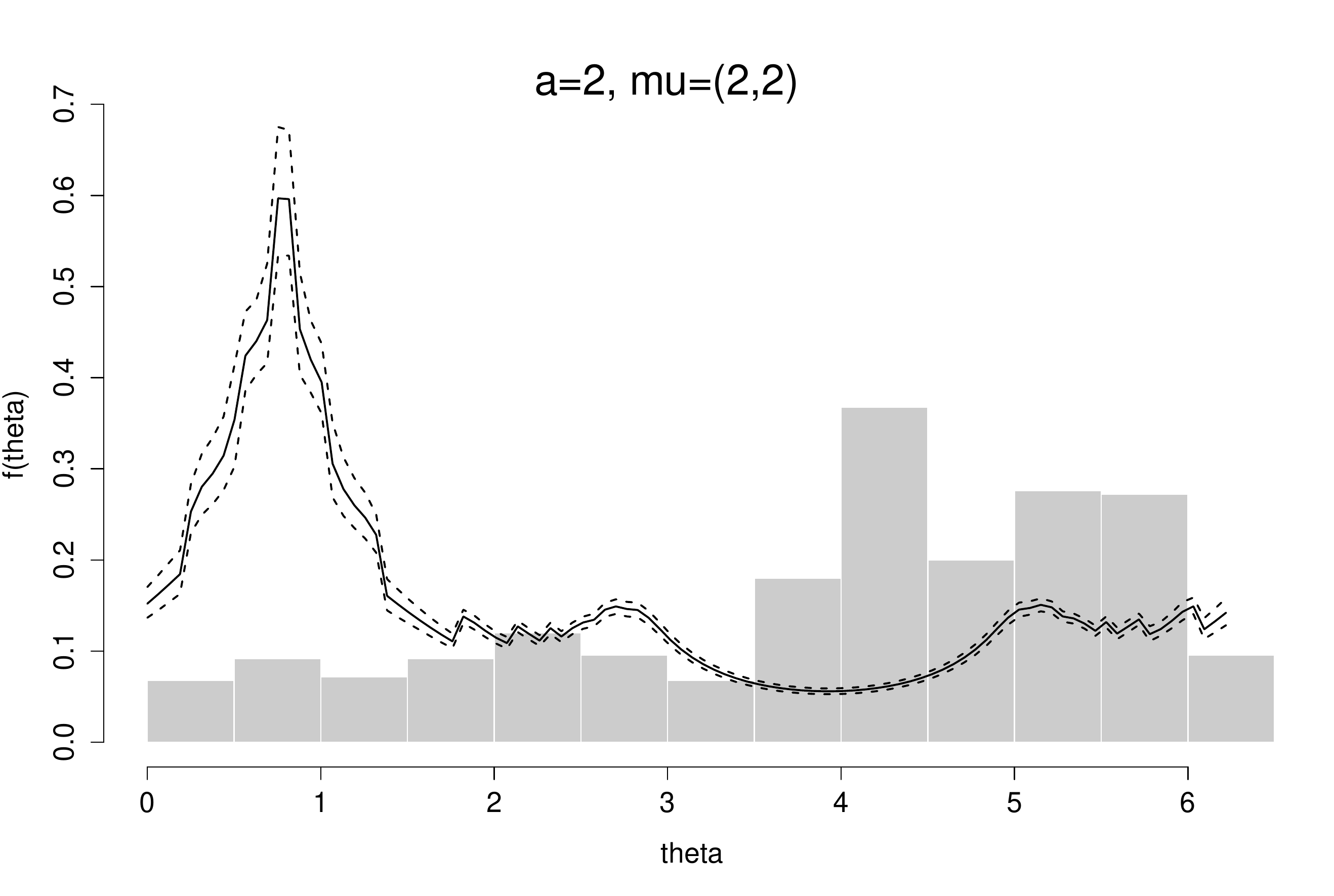}}
\caption{{\small Posterior density estimates for simulated data with $n=500$. Across columns $\alpha=0.5$ and $\alpha=2$. Across rows $\bmu'=(0,0)$, $\bmu'=(1,1)$ and $\bmu'=(2,2)$.}}
\label{fig:postsim1}
\end{figure}

\begin{figure}
\centerline{\includegraphics[scale=0.31]{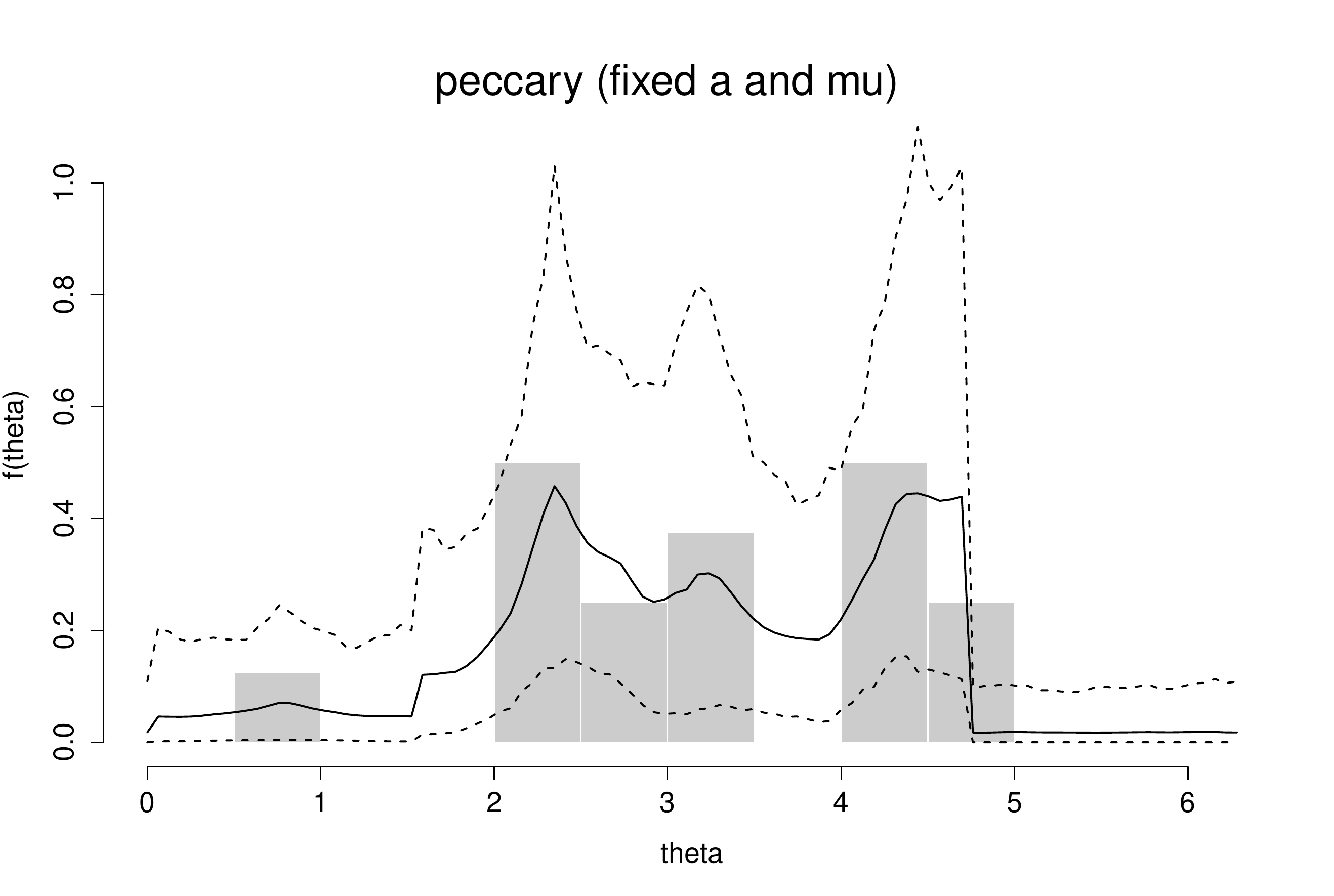}\includegraphics[scale=0.31]{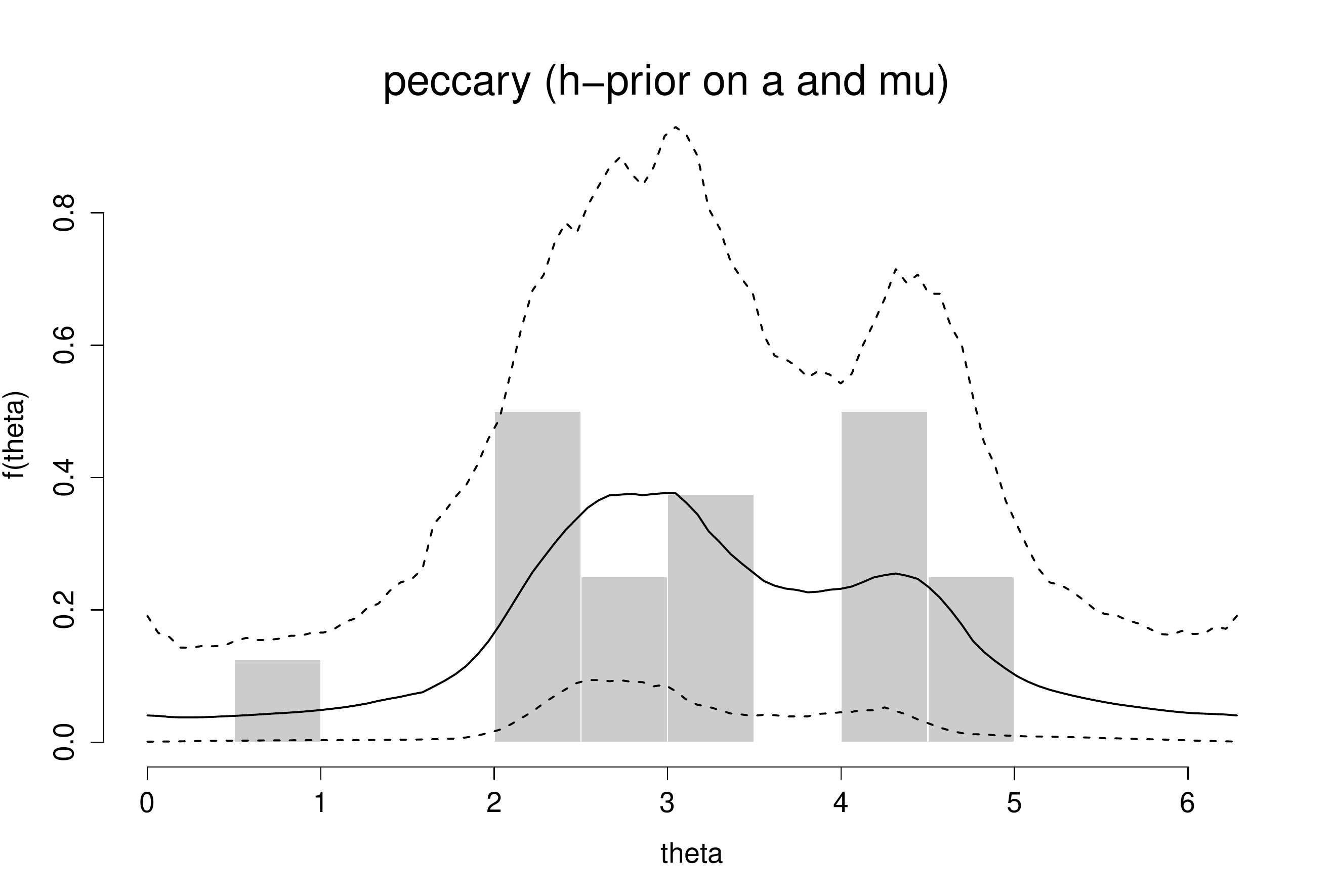}}
\centerline{\includegraphics[scale=0.31]{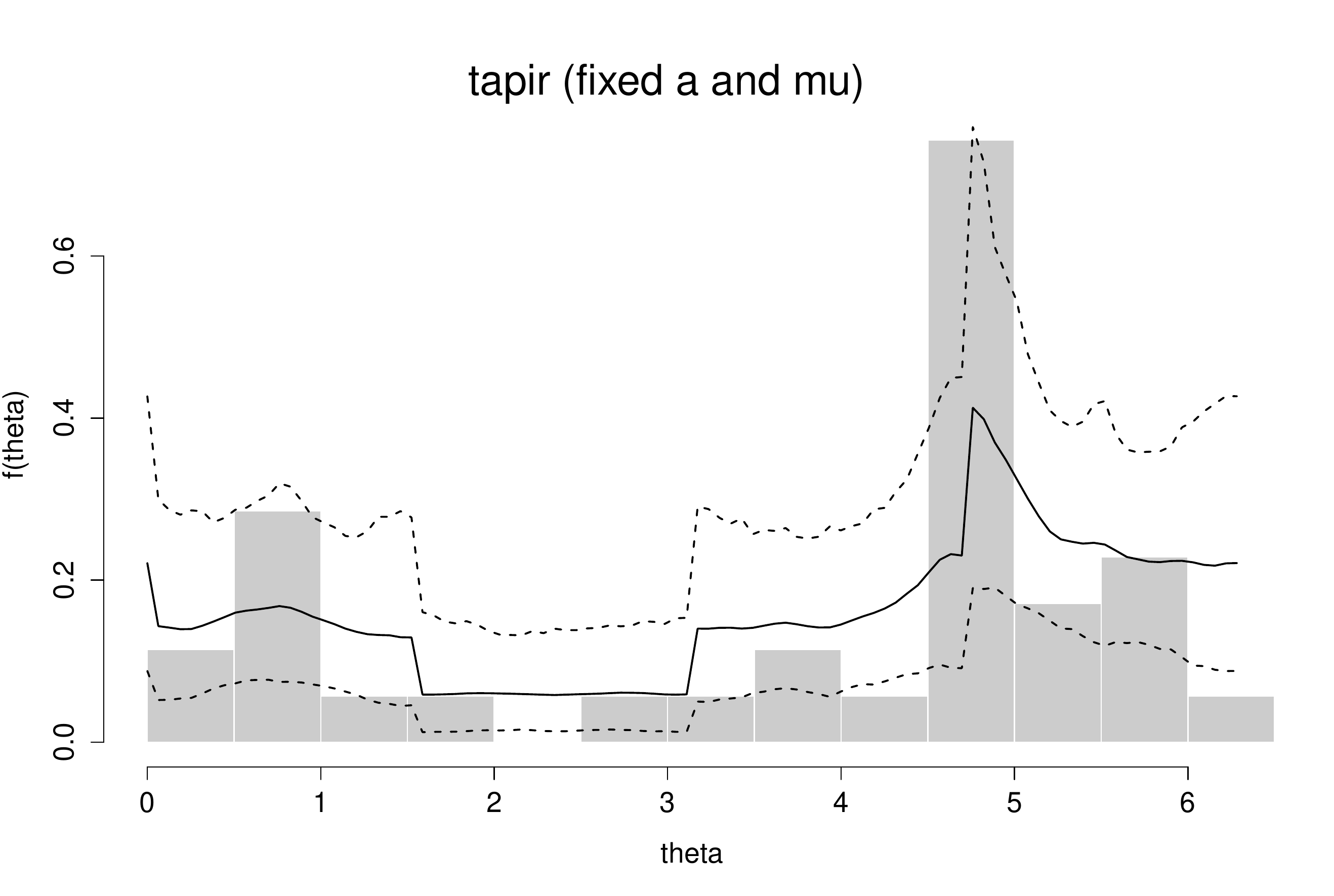}\includegraphics[scale=0.31]{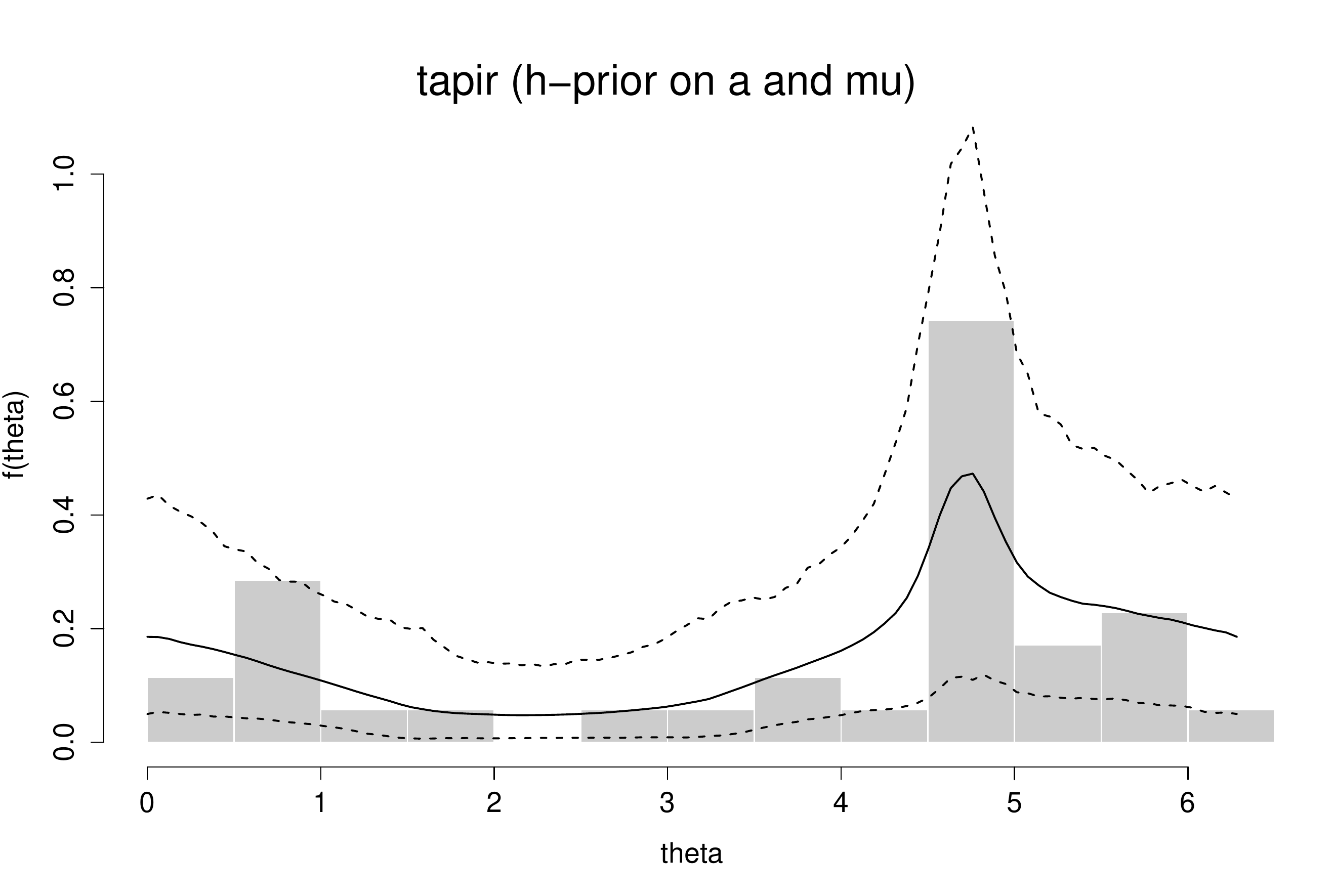}}
\centerline{\includegraphics[scale=0.31]{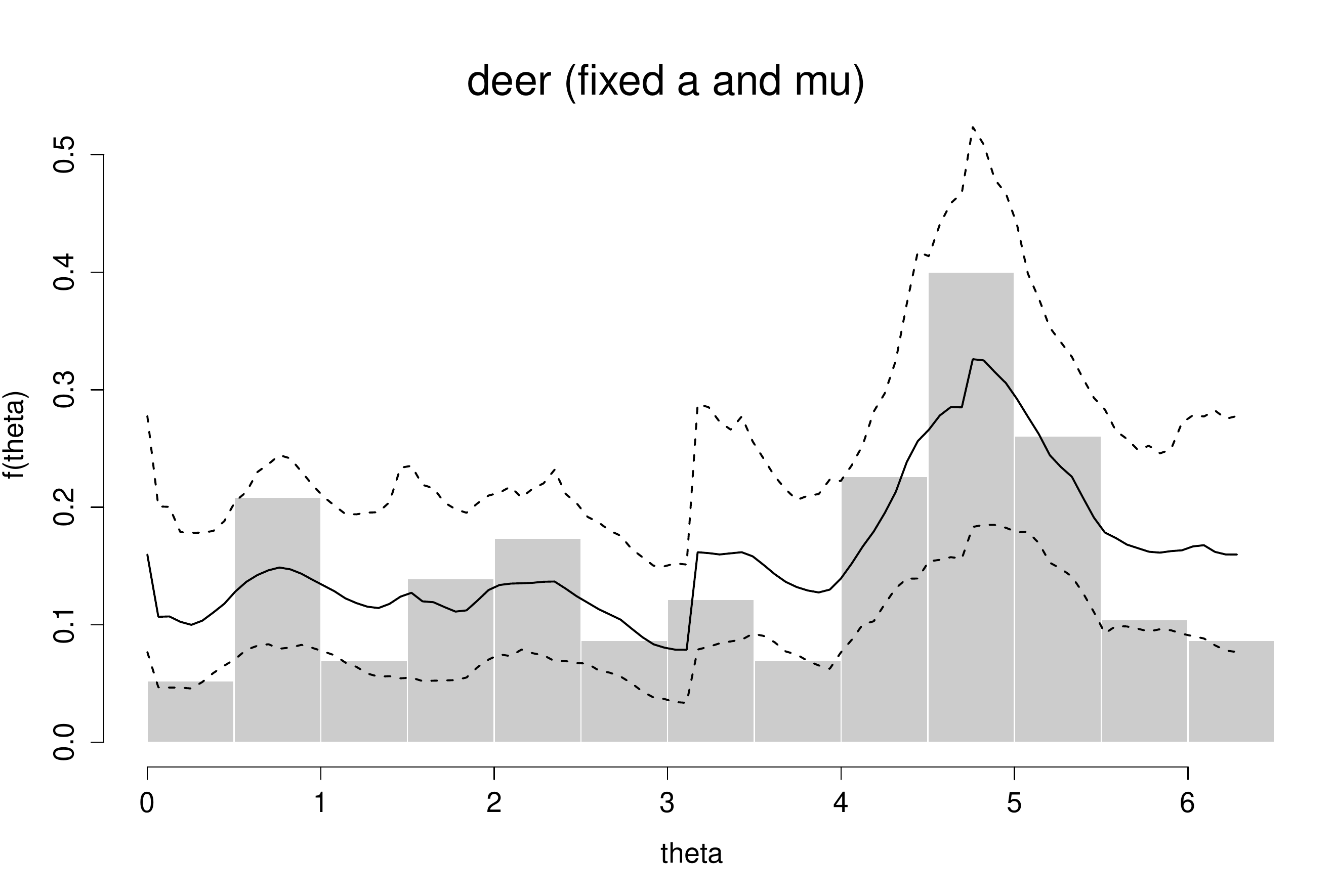}
\includegraphics[scale=0.31]{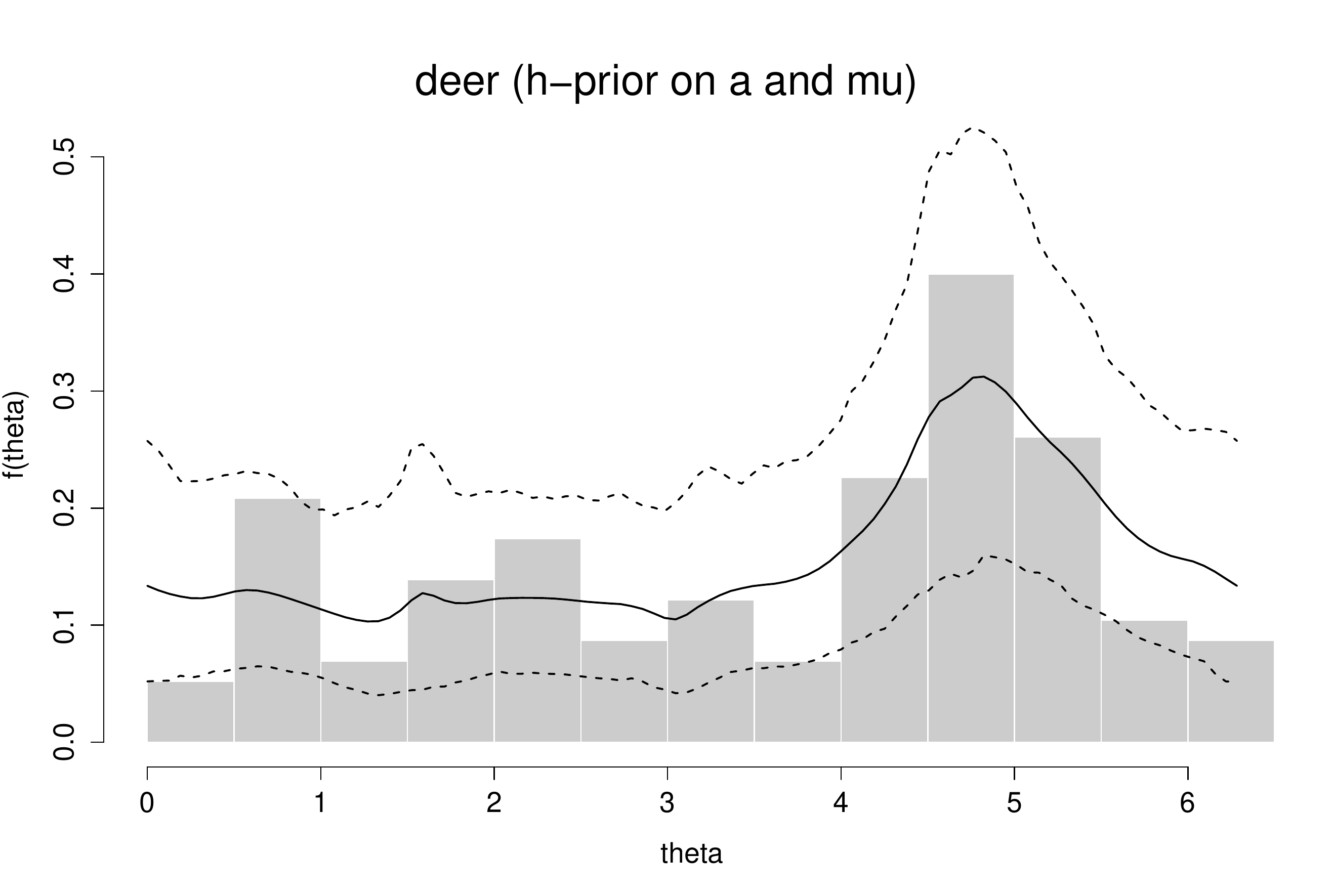}}
\caption{{\small Posterior density estimates for the temporal activity of three animals from El Triunfo Reserve.}}
\label{fig:postreal}
\end{figure}

\begin{figure}
\centerline{\includegraphics[scale=0.65]{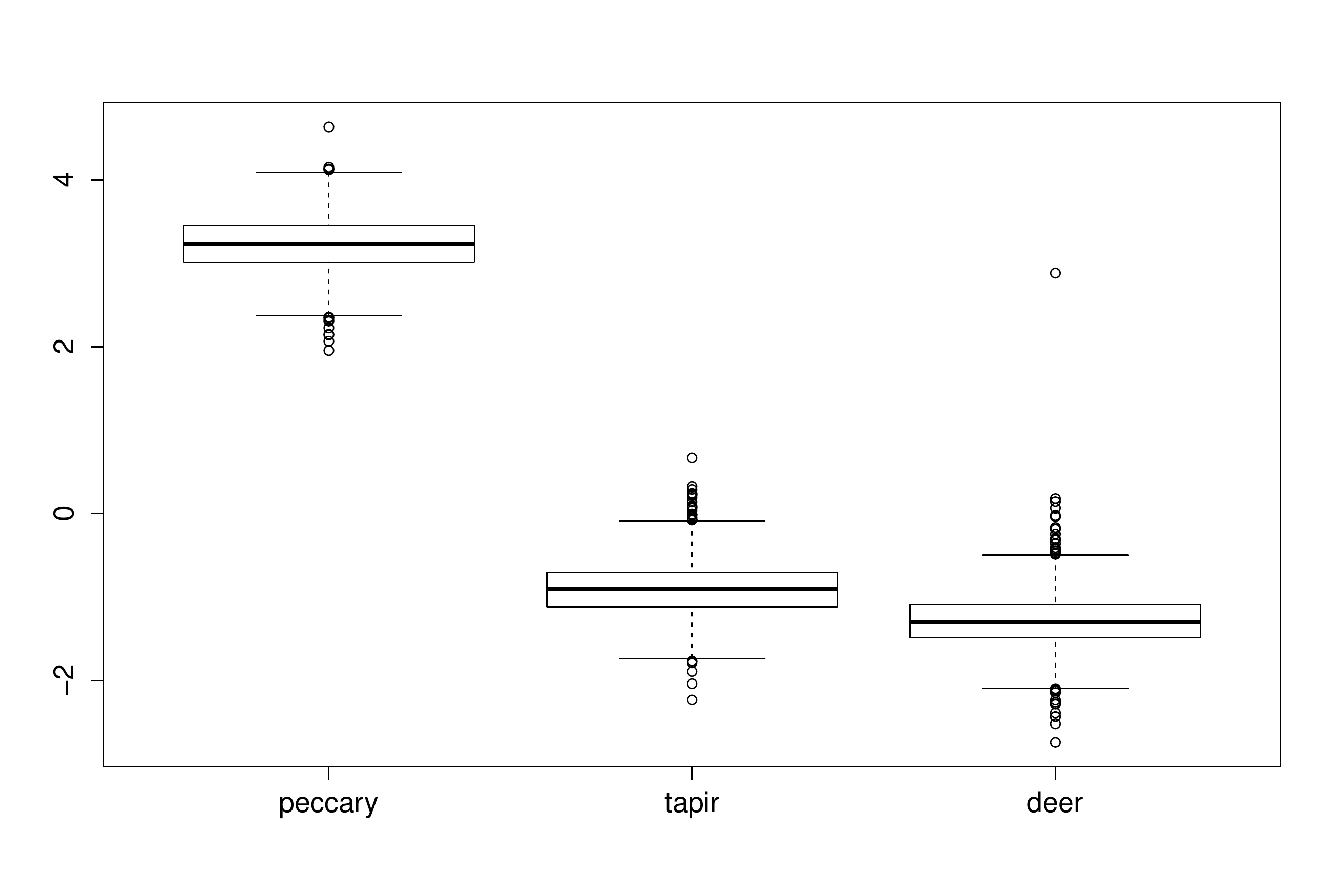}}
\caption{{\small Posterior distribution (boxplot) of the mean time of activity, $\nu_\theta$, for the three animals from El Triunfo Reserve.}}
\label{fig:postmeans}
\end{figure}

\end{document}